\documentclass[journal]{IEEEtran}

\pdfoutput=1
\usepackage[utf8]{inputenc}
\usepackage[pdftex]{graphicx}
\graphicspath{{Figures/}}
\usepackage{subfig}
\usepackage{amsmath, amssymb, amsfonts, bbm}
\usepackage{epstopdf} 
\usepackage{float} 

\usepackage{enumitem}
\usepackage{booktabs}
\usepackage{cite}

\usepackage{caption}
\usepackage{amsthm} 
    
    \newtheorem{assumption}{Assumption}

    \newtheorem{remark}{Remark}
    \newtheorem{theorem}{Theorem}

\usepackage[usenames,dvipsnames]{xcolor}
\usepackage{hyperref}
\hypersetup{%
    colorlinks=true,
    linkcolor={blue!50!black},
    citecolor={blue!50!black},
    urlcolor={blue!60!black}
}

\usepackage{fancyhdr}

\newcommand{\p}{\mathrm{p}}
\newcommand{\pr}{\mathrm{pr}}
\renewcommand{\r}{\mathrm{r}}
\newcommand{\st}{\mathop{\rm s.t.}\nolimits}

\newcommand{\R}{\mathbb{R}}
\newcommand{\N}{\mathbb{N}}
\newcommand{\UU}{\mathcal{U}}
\newcommand{\YY}{\mathcal{Y}}
\newcommand{\XX}{\mathcal{X}}
\newcommand{\open}[1]{{\rm int}(#1)}
\newcommand{\T}{^\top}
\newcommand{\matrice}[2]{\left[\hspace*{-.1cm}\begin{array}{#1} #2 \end{array}\hspace*{-.1cm}\right]}

\def\moveEq#1{{}\mkern#1mu} 
\def\vv#1{{ \rm \bf{#1}}} 

\def\nx{{n_x}}
\def\nxp{{n_{x_\p}}}
\def\nu{{n_u}}
\def\ny{{p}}
\def\nd{{n_d}}
\def\ndx{{n_{d_x}}}
\def\ndy{{n_{d_y}}}
\def\nth{{n_\theta}}

\begin{document}
\title{\LARGE \bf Learning disturbance models for offset-free reference tracking}%

\author{Pablo~Krupa,~Mario~Zanon,~Alberto~Bemporad%
\thanks{This work has received support from the European Research Council (ERC), Advanced Research Grant COMPACT (Grant Agreement No. 101141351).
The authors are with the IMT School for Advanced Studies, Lucca, Italy. Emails: {\tt\small \{pablo.krupa, mario.zanon, alberto.bemporad\}@imtlucca.it}. Corresponding author: Pablo Krupa.
The authors wish to thank Vittorio Mattei for initial discussions on the topic of this paper.%
}}%

\pagestyle{fancy}
\maketitle
\thispagestyle{fancy}

\begin{abstract}
This work presents a nonlinear control framework that guarantees asymptotic offset-free tracking of generic reference trajectories by learning a nonlinear disturbance model, which compensates for input disturbances and model-plant mismatch.
Our approach generalizes the well-established method of using an observer to estimate a constant disturbance to allow tracking constant setpoints with zero steady-state error.
In this paper, the disturbance model is generalized to a nonlinear static function of the plant's state and command input, learned online, so as to perfectly track time-varying reference trajectories under certain assumptions on the model and provided that future reference samples are available.
We compare our approach with the classical constant disturbance model in numerical simulations, showing its superiority.
\end{abstract}

\begin{IEEEkeywords}
    Offset-free reference tracking, nonlinear model predictive control, extended Kalman filter, disturbance model
\end{IEEEkeywords}

\section{Introduction} \label{sec:introduction}

Control techniques can be grouped in two main categories: model-based and model-free techniques.
The latter ones can achieve the control objective without necessarily taking advantage or needing a model of the system. Examples can be PID control \cite{chong2005} and model-free reinforcement learning \cite{recht2019}.
However, these methods have some limitations, such as the difficulty to guarantee stability or safety with respect to given constraints without a model.
This motivates the introduction of model-based techniques.
Among them, Model Predictive Control (MPC) is a well-known optimization-based technique \cite{RMD17, borrelli2017}.
At each sample time, MPC solves a finite-horizon optimal control problem using a prediction model and the new measurement (or estimate) of the system state.
In MPC, offsets with respect to the desired reference might occur due to mismatches between the prediction model and the (unknown) system dynamics, or due to unmeasured disturbances.

Offset-free MPC schemes are used to reject both model mismatch and unknown disturbances, leading, as the name suggests, to offset-free tracking of the desired reference.
Most of these schemes rely on using an augmented state-disturbance model (see \cite{pannocchia2015_linear} for other formulations), leading to an augmented state that is estimated in order to reject the disturbance.
Several works have been done in the field of linear and nonlinear MPC to attain offset-free tracking of piecewise-constant references, leading to an established theory; see \cite{pannocchia2015_linear, pannocchia2015, pannocchia2007, morari2012, betti2013}.
However, new results in this field are still abundant.
For instance, in \cite{bonassi2021} the authors use a gated recurrent unit neural network to identify the system and use it as a prediction model for achieving offset-free tracking of constant references.
In \cite{bonassi2022}, the authors use a NARX neural network for the same purposes as \cite{bonassi2021}.
In \cite{son2022}, an artificial neural network is used to model the disturbances at steady-state, but in the context of linear MPC.
In \cite{caspari2021}, the authors propose an approach for the generation of a disturbance model by taking advantage of sufficient observability conditions and solve a semi-infinite program offline in order to retrieve a disturbance model that can be used online for offset-free nonlinear MPC. 
In \cite{mirasierra_IFAC_2020}, the authors propose a modifier-adaptation approach that achieves offset-free tracking for periodic reference trajectories.

In this paper, we present theoretical conditions under which offset-free tracking of time-varying reference trajectories can be achieved by exploiting a nonlinear disturbance model in the context of nonlinear control. 
We extend previous results in offset-free tracking into a more general form, where we assume to have a preview of future reference signals in order to also achieve offset-free tracking of non-constant reference trajectories; in contrast to previous results, which consider constant references or periodic reference trajectories, cf.~\cite{pannocchia2015_linear, pannocchia2015, pannocchia2007, morari2012, betti2013, bonassi2021, bonassi2022, son2022, caspari2021, mirasierra_IFAC_2020}.
We establish theoretical conditions for offset-free tracking using a nonlinear static disturbance model, which is learned online.
Thus, we can see this approach as ``grey-box", since it mixes an offline white-box state-space model of the system with the online estimation of a set of parameters for the nonlinear disturbance function.
Although the theoretical assumptions for offset-free tracking may be restrictive in a general case, we show how the proposed theoretical framework can provide good results when applied to the frequently used combination of a nonlinear MPC with an Extended Kalman Filter (EKF) as a state observer.
In particular, we present numerical results that show how this setup can be used to train a disturbance model online, leading to offset-free tracking of generic reference trajectories if the disturbance model satisfies the required theoretical assumptions.
Furthermore, we show results using a recurrent neural network as the disturbance model, which we train online following a similar approach to~\cite{bemporad_TAC_2023}.
The experiments indicate that, even when the theoretical assumptions are not fully satisfied, the proposed disturbance model can outperform the classical constant disturbance model~\cite{pannocchia2015_linear, morari2012}, in the context of nonlinear offset-free MPC.

The paper is organized as follows.
In Section~\ref{sec:problem} we present the problem formulation along with the assumptions and theoretical results that lead to offset-free tracking of time-varying reference trajectories.
We then show now the combination of nonlinear MPC and EKF can be applied to this control objective in Section~\ref{sec:NMPC_EKF}.
In particular, we show that the EKF can be used to learn the parameters of the disturbance model online, achieving very good tracking results even when the theoretical assumptions presented in Section~\ref{sec:problem} are not fully satisfied.
We conclude with Section~\ref{sec:no_guarantees}.

\noindent\textbf{Notation:~}
Given a set $S\subseteq \R^n$, $\open{S}$ denotes its interior.
$x(k|j) \in \R^n$ is the estimate of $x \in \R^n$ at time $k \in \N$ given the information available at time $j \in \N$.
The natural numbers $\N$ include $0$, and $\N_i^j \doteq \{i, i+1, \dots j-1, j\}$.
Given $x \in \R^n$ and some positive semidefinite matrix $Q \in \R^{n \times n}$, $\| x \|^2_Q = x\T Q x$.

\section{Offset-free reference tracking} \label{sec:problem}

We consider the problem of achieving offset-free reference tracking of an unknown nonlinear system.
We establish sufficient conditions under which a nonlinear controller and observer can achieve this goal using a nonlinear nominal model augmented with a suitably selected disturbance~model.

\subsection{The real system}
We assume that the controlled process is described by the discrete-time nonlinear dynamics
\begin{equation} \label{eq:plant}
	\begin{split}
		x_\p(k+1) & = f_\p(x_\p(k),u(k)),\\
		y_\p(k) & = g_\p (x_\p (k)),
	\end{split}
\end{equation}
where $x_\p \in \R^\nxp$, $u \in \R^\nu$ and $y_\p \in \R^\ny$, denote, respectively, the state, input and output of the process
and $k \in \N$ is the sampling instant.
We assume that we do not know the functions $f_\p \colon \R^\nxp \times \R^\nu \to \R^\nxp$ and $g_\p \colon \R^\nxp \to \R^\ny$. We also assume
that only input and output measurements are available, i.e., we cannot directly access the state vector $x_\p(k)$, whose dimension $\nxp$ may also be unknown.

Our aim is to design a controller that makes the output $y_\p(k)$ of plant~\eqref{eq:plant}
track a given generic reference signal $\{r(k)\}_{k=0}^\infty$, under the following
input and output constraints
\begin{equation} \label{eq:constraints}
    u(k)\in\UU,\quad y_\p(k)\in\YY,\quad \forall k \in \N,
\end{equation}
where $\UU \subseteq \R^\nu$, $\YY \subseteq \R^\ny$ are nonempty.
Since the objective is to achieve offset-free tracking, we make the following two standing assumptions, which are obvious requirements to be able to track the time-varying reference $r(k)$ asymptotically with zero error while satisfying the system constraints.

\begin{assumption} \label{ass:trackability}
The reference signal $\{r(k)\}_{k=0}^\infty$ satisfies $r(k)\in \open{\YY}$, $\forall k \in \N$.
Furthermore, there exist trajectories $\{x_{\pr}(k)\}_{k=0}^\infty$ and $\{u_\r(k)\}_{k=0}^\infty$ such that 
\begin{subequations}
\begin{align}
    x_{\pr}(k+1)&=f_\p(x_{\pr}(k),u_\r(k)),\\
    r(k)&=g_\p(x_{\pr}(k)),
\end{align}
\label{eq:xpr}%
\end{subequations}
and $u_\r(k)\in\UU$ for all $k \in \N$.
\end{assumption}

\begin{assumption} \label{ass:stabilizability}
Let $\{r(k)\}_{k=0}^\infty$ satisfy Assumption~\ref{ass:trackability} and $\{x_{\pr}(k)\}_{k=0}^\infty$, $\{u_\r(k)\}_{k=0}^\infty$ denote a corresponding pair of state and input trajectories.
There exists a nonempty set $\XX_{\p}^0 \subseteq \R^\nxp$ of initial states such that for each $x_\p(0) \in \XX_{\p}^0$ there exists an input trajectory $\{u(k)\}_{k=0}^\infty$, $u(k)\in\UU$, for which \eqref{eq:plant} satisfies $\lim\limits_{k\rightarrow\infty} x_\p(k)-x_{\pr}(k) = 0$ and $y_\p(k)\in\open{\YY}$ for all $k \in \N$.
\end{assumption}

Assumption~\ref{ass:trackability} is a necessary condition for perfect tracking under strict feasibility of the corresponding input and output trajectories.
The assumption requires the reference trajectory $r(k)$ to be a feasible reference for the real system dynamics~\eqref{eq:plant} and constraints~\eqref{eq:constraints}.
Since the real system dynamics are unknown, this might be difficult to verify in general.
However, in many practical settings it may be possible to verify Assumption~\ref{ass:trackability} due to sufficient knowledge of the real system or to historical data.
Assumption~\ref{ass:stabilizability} is an unrestrictive assumption stating that there exist some set of initial states from which the system can be asymptotically steered to the reference.
In practice this is not restrictive, since it is commonly assumed that systems are controlled starting from a viable initial state that is close enough to the desired reference.

\subsection{Control-oriented model and estimation}

As is often the case in model-based controllers, we consider a nominal prediction model with disturbance $d \in \R^\nd$:
\begin{equation} \label{eq:model}
	\begin{aligned}
		x(k+1) &= f(x(k),u(k),d(k)), \\
        d(k) &= h(x(k),u(k), \theta(k)), \\
		y(k) &= g(x(k),d(k)),
    \end{aligned}
\end{equation}
where $x \in \R^\nx$ is the state of the prediction model, $u \in \R^\nu$ and $y \in \R^\ny$ are, respectively, the input and the output of~\eqref{eq:plant}, and
$f \colon \R^\nx \times \R^\nu \times \R^\nd \to \R^\nx$, $g \colon \R^\nx \times \R^\nd \to \R^\ny$ are parametrized by $\theta \in \R^\nth$.
Vector $d\in\R^\nd$ is a disturbance that affects the model and is generated by a parametric disturbance function $h \colon \R^\nx \times \R^\nu \times \R^\nth \to \R^\nd$ whose parameters $\theta$ are estimated online.
We note that we make no assumption on the dimension $\nx$ with respect to the dimension $\nxp$.

Model~\eqref{eq:model} is the combination of a nominal model of the system $f$, $g$, obtained offline, and the disturbance function $h$ of the parameter vector $\theta$ that is learned online. 
We note that the use of a prediction model learned offline is the typical approach in model-based controllers, such as MPC~\cite{RMD17}.
Functions $f$ and $g$ may be obtained, for instance, from system identification, leading to a prediction model that balances prediction accuracy and complexity.
A consequence of this balance is a mismatch between the prediction model and the real system dynamics~\eqref{eq:plant}, which generally leads to offset when tracking a given reference.
In offset-free MPC for piecewise-constant references, this is solved by augmenting the nominal prediction model $f$, $g$, with a disturbance $d \in \R^\nd$ that is estimated online, see, e.g.,~\cite[\S 2.3]{pannocchia2015}.
In this paper, we augment the nominal prediction model $f$, $g$, with a disturbance function $h$ whose parameters $\theta$ will be learned online using an estimator to achieve offset-free tracking.
The approach is analogous to the one used in offset-free MPC for piecewise-constant references, but with the additional complexity of $h$ required to achieve this goal for time-varying references.

Model~\eqref{eq:model} enables a trade-off between offline identification, which captures the most relevant plant dynamics, and online model adaptation to cope with unknown disturbances and model-plant mismatches.
For instance, if a nominal model 
\begin{equation*}
		x(k+1) = f_n(x(k),u(k)), \quad y(k) = g_n(x(k)),
\end{equation*}
is obtained from system identification, then $f$ and $g$ can be obtained as $f = f_n + h_x$ and $g = g_n + h_y$, with suitably selected disturbance functions $h_x(\cdot)$, $h_y(\cdot)$.
We consider model~\eqref{eq:model} to provide a more general problem setup.

In order to achieve offset-free tracking, we consider the following two assumptions on model~\eqref{eq:model}.

\begin{assumption} \label{ass:continuity}
Function $g$ is continuous, function $f$ is continuous with respect to $(x,d)$, and function $h$ is continuous with respect to $(x,\theta)$.
\end{assumption}

\begin{assumption} \label{ass:perfect-modeling}
Let $\{r(k)\}_{k=0}^\infty$ satisfy Assumption~\ref{ass:trackability}
and $\{u_\r(k)\}_{k=0}^\infty$ denote a corresponding
input trajectory. 
There exists $\theta_\r\in\R^\nth$ and a state trajectory $\{x_\r(k)\}_{k=0}^\infty$ such that
\begin{equation} \label{eq:perfect-tracking}
    \begin{aligned}
    x_\r(k+1)&=f(x_\r(k),u_\r(k),d_\r(k)), \\
    d_\r(k)&=h(x_\r(k),u_\r(k),\theta_\r), \\
    r(k)&=g(x_\r(k),d_\r(k)), \; \forall k \in \N.
    \end{aligned}
\end{equation}
\end{assumption}

Assumption~\ref{ass:continuity} is a mild technical requirement for the proofs reported in the sequel.
We note that Assumption~\ref{ass:continuity} might be insufficient for many controllers and/or observers from the literature, which might require additional assumptions on the functions of model~\eqref{eq:model}.
Assumption~\ref{ass:perfect-modeling} guarantees that, limited to perfect tracking conditions, model~\eqref{eq:model} is versatile enough to reproduce the output reference signal $r(k)$ when excited by the same associated reference input $u_\r(k)$. 
Note that in the case of constant references $r(k)\equiv\bar r$ if $u_\r(k)\equiv\bar u_\r$ and $\bar x_{\pr}$ are such that $\bar x_{\pr}=f_\p(\bar x_{\pr},\bar u_\r)$, $\bar r=g_\p(\bar x_{\pr})$, then, Assumption~\ref{ass:perfect-modeling} always holds for the classical additive output disturbance model \cite{pannocchia2015_linear}
    $x(k+1)=f(x(k),u(k))$,
    $d(k)=\theta(k)$,
    $y(k)=g(x(k))+d(k)$,
as any $\bar x_\r$, $\bar \theta_\r$ such that $\bar x_\r=f(\bar x_\r,\bar u_\r)$ and $\bar \theta_\r=\bar r-g(\bar x_\r)$ makes~\eqref{eq:perfect-tracking} be satisfied, where clearly $\bar\theta_\r$ represents a term to correct the plant/model mismatch of the output vector at steady-state.
For this reason, Assumption~\ref{ass:perfect-modeling} is usually not explicitly reported in the literature on offset-free MPC, while we need to introduce it here to handle the more general time-varying reference setting.
Choosing a disturbance model $h(\cdot)$ such that Assumption~\ref{ass:perfect-modeling} is known to be satisfied may not be possible in many practical settings due to a lack of knowledge of the real system dynamics.
However, as we illustrate in the numerical results of Section~\ref{sec:results}, the use of a general-purpose disturbance model can provide good tracking performance, even if Assumption~\ref{ass:perfect-modeling} is not fully satisfied.

To estimate the state $x(k)$ and parameters $\theta(k)$ online, we rely on an observer that
delivers the estimates
\begin{subequations} \label{eq:observer}
\begin{equation} \label{eq:xkk}
	\matrice{c}{x(k|k) \\ \theta(k|k)} = \matrice{c}{x(k|k-1) \\ \theta(k|k-1)} + \omega(k,e(k)),
\end{equation}
based on the output prediction error
\begin{align} \label{eq:output_error}
	e(k) \doteq y_\p(k) - g(x(k|k-1), d(k|k-1)),
\end{align}
where the measurement-update function $\omega \colon \N \times\R^\ny \to \R^{\nx + \nth}$ provides the correction term due to the output prediction error.
We assume that the time-update function of the observer is
\begin{align}
	d(k|k-1) &= h(x(k|k-1),u(k),\theta(k|k-1)), \\
	d(k|k) &= h(x(k|k),u(k),\theta(k|k)), \\
    \matrice{c}{x(k+1|k) \\ \theta(k+1|k)} &= \matrice{c}{f(x(k|k), u(k), d(k|k)) \\ \theta(k|k)}.
\end{align}%
\end{subequations}

The following assumption is a standard requirement for any well-posed observer, cf., e.g.,~\cite[Assumption 9]{pannocchia2015}.

\begin{assumption} \label{ass:observer}
	The observer-update function $\omega$ satisfies $\omega(k,0)=0$, for all $k\in\mathbb{N}$, and $\omega$ is continuous with respect to its second argument in a neighborhood of the origin. 
\end{assumption}

\begin{remark}[Comparison with constant disturbance models]
The case of constant disturbance models (see, e.g.,~\cite{pannocchia2015})
\begin{equation}
\begin{aligned}
	x(k+1)&=F(x(k),u(k),d(k)),\; y(k)=G(x(k),d(k)), \\
	d(k+1)&=d(k),
\end{aligned}
\label{eq:constant-dist}
\end{equation}
is a special case of the general disturbance model~\eqref{eq:model}, obtained by setting  $\nth=\nd$, $f(x,u,d)=F(x,u,d)$, $g(x,d)=G(x,d)$, $d(k) = \theta(k)$ and $h(x,u,\theta)=\theta$. 
Vice versa, given the model in~\eqref{eq:model}, one can obtain the model in~\eqref{eq:constant-dist} by setting $\nd=\nth$, $d(k)=\theta(k)$, $F(x,u,d)=f(x,u,h(x,u,d))$, and $G(x,d)=g(x,h(x,u,d))$. 
It is therefore apparent that the two disturbance
modeling frameworks are mathematically equivalent. However, our framework has the advantage of being more structured, as it explicitly models the disturbance vector as the output of a parametric nonlinear model of the state and input vectors, with $\theta$ being the vector of disturbance model parameters, while in~\eqref{eq:constant-dist}, the underlying modeling assumption is that the disturbance is an unknown constant, which is fully justified by the fact that the emphasis is on compensating steady-state errors when tracking constant set-points.
We remark also that in~\cite{pannocchia2015} the authors assume that the number $\nx$ of model states is equal to the order $\nxp$ of the plant state $x_\p$, while in this paper we do not make such an assumption.
We do not even assume that $\nxp$ is known.
\end{remark}

\subsection{Nonlinear controller}
We assume that a nonlinear controller 
\begin{equation} \label{eq:nlcon}
    u(k)= \kappa (x(k),r(k),\theta(k)),
\end{equation}
$\kappa \colon \R^\nx \times\R^\ny \times \R^\nth \to\UU$, has been designed fulfilling the following assumption.

\begin{assumption} \label{ass:contr}
Consider any $\{r(k)\}_{k=0}^\infty$ satisfying Assumption~\ref{ass:trackability} and any constant value $\theta(k) \equiv \bar \theta$.
Let $\{\Delta x(k),\Delta\theta(k)\}$ be vanishing perturbations, i.e., 
\begin{equation*}
    \lim_{k\rightarrow\infty}\Delta x(k) =0, \quad \lim_{k\rightarrow\infty}\Delta \theta(k) =0.
\end{equation*}
Then, applying $u(k)= \kappa (x(k)+\Delta x(k), r(k), \theta(k)+\Delta\theta(k))$ makes the output of the nominal model~\eqref{eq:model} track the reference asymptotically without errors, i.e.,
\begin{equation} \label{eq:no-ss-error}
    \lim_{k\rightarrow\infty}g(x(k),d(k))-r(k)=0,
\end{equation}
with $y(k) \in \open{\YY}$ for all $k \in \N$.
\end{assumption}

Note that Assumption~\ref{ass:contr} implies that $u(k)\in\UU$ by definition, and that the closed-loop system constituted by~\eqref{eq:model} and~\eqref{eq:nlcon} is intrinsically robust to vanishing state perturbations affecting both the evolution of the nominal model and the state-feedback signals to the controller.
Such a view of the actual control system is depicted in Fig.~\ref{fig:control_loop}, which shows our reinterpretation of the control system as a state-feedback loop with state $(x(k|k-1), \theta(k|k-1))$, output $y(k|k)$ and input $u(k)$ generated by the controller~\eqref{eq:nlcon} under the feedback perturbations 
\begin{equation}
    \matrice{c}{\Delta x(k)\\ \Delta\theta(k)} = \omega(k,e(k)),
\label{eq:Delta-x-th}
\end{equation}
caused by the real plant through the measurement-update mapping~\eqref{eq:xkk} of the observer.
Clearly, when $e(k)=0$ the measurement update block becomes an all-pass filter having no effect on the control loop, which recovers its nominal behavior given by the evolution of the model equations~\eqref{eq:model} under the input $u(k)$ generated by~\eqref{eq:nlcon}.
In other words, the dynamics of the plant become completely irrelevant in the way the state observer and controller evolve when $e(k)=0$.
Indeed, the following theorem shows that if $e(k)$ vanishes asymptotically then the output $y_\p(k)$ of plant~\eqref{eq:plant} perfectly tracks $r(k)$.

\begin{figure}[t]
\centering
    \includegraphics[width=0.95\columnwidth]{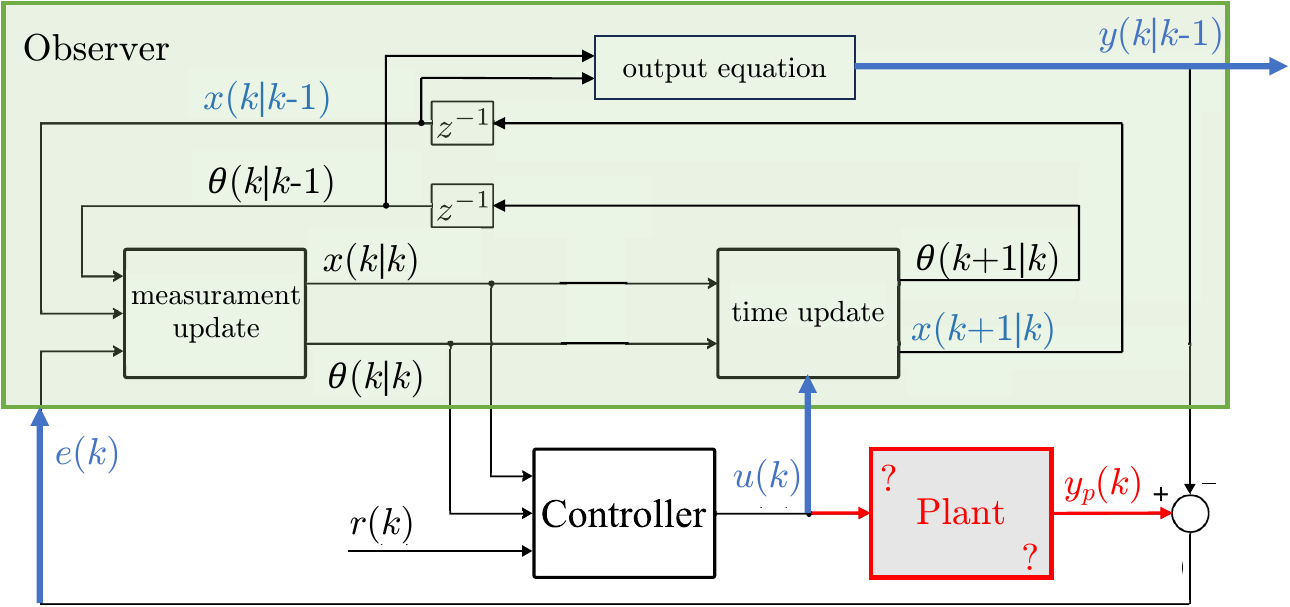}
    \caption{Interpretation of the closed-loop system from the point of view of the nominal model, where the plant is seen as a generator of the ``disturbance" $e(k)$ whose effect is rejected by the observer thanks to the disturbance model.}
    \label{fig:control_loop}
\end{figure}

\begin{theorem} \label{th:main}
Consider the closed-loop system constituted by~\eqref{eq:plant} under the control law $u(k) = \kappa (x(k|k),r(k),\theta(k|k))$, where $x(k|k)$ and $\theta(k|k)$ are obtained by~\eqref{eq:observer}.
Let Assumptions~\ref{ass:trackability}--\ref{ass:contr} hold. 
Then, for any initial plant state $x_\p(0)\in\XX_{\p}^0$, convergence of the nonlinear observer~\eqref{eq:observer}, i.e., 
\begin{equation}
	\lim_{k\rightarrow\infty} e(k)=0,
\label{eq:e(k)_asympt}
\end{equation}
implies asymptotic perfect tracking
\begin{equation}
	\lim_{k\rightarrow\infty} y_\p(k)-r(k)=0.
\label{eq:y(k)_asympt}
\end{equation}
Moreover, $u(k)\in\UU$ for all $k\geq 0$ and there exists a time index $k_f\geq 0$ such that $y_\p(k)\in\YY$ for all $k\geq k_f$.
\end{theorem}

\begin{proof}
Consider the perturbed dynamical system
\begin{align*}
    x(k{+}1|k) &= f(x(k|k{-}1) {+} \Delta x(k),u(k),d(k|k{-}1) {+} \Delta d(k)), \\
    \theta(k{+}1|k) &= \theta(k|k{-}1) {+} \Delta \theta(k), \\
    u(k) &= \kappa (x(k|k{-}1) {+} \Delta x(k),r(k), \theta(k|k{-}1) {+} \Delta\theta(k)),
\end{align*}
where $\Delta x(k)$, $\Delta \theta(k)$ are given by~\eqref{eq:Delta-x-th} and
    $\Delta d(k) = h(x(k|k-1) + \Delta x(k),u(k), \theta(k|k-1)+\Delta\theta(k))
                - d(k|k-1)$.
Since, by Assumption~\ref{ass:observer}, the observer feedback $\omega$ is continuous with respect to the output estimation error in a neighborhood of the origin and $\lim_{k\rightarrow\infty} e(k)=0$, then from~\eqref{eq:xkk} we have that the perturbations $\Delta x(k) \rightarrow 0$ and $\Delta\theta(k)\rightarrow 0$ as $k \rightarrow \infty$.
By Assumption~\ref{ass:continuity}, function $h$ is continuous with respect to $x$ and $d$, so that also $\Delta d(k)\rightarrow 0$.
Hence, by Assumptions~\ref{ass:perfect-modeling} and~\ref{ass:contr}, we have that $g(x(k|k-1), d(k|k-1))-r(k) \rightarrow 0$ and, finally, since $y_\p(k)- y(k|k-1)=e(k)\rightarrow 0$, also that the \emph{actual} tracking error
\begin{align*}
    y_\p(k){-}r(k) &= y_\p(k)-y(k|k-1)+y(k|k-1)-r(k) \\
                   &= e(k)+g(x(k|k-1),d(k|k-1))-r(k)\rightarrow 0.
\end{align*}
Since, by Assumption~\ref{ass:trackability}, $r(k)\in\open{\YY}$ and $y_\p(k)-r(k)\rightarrow 0$, then there exists $k_f \in \N$ such that $y_\p(k)\in\YY$, $\forall k\geq k_f$. \qedhere
\end{proof}

Note that Assumption~\ref{ass:stabilizability} has not been explicitly mentioned in the proof of Theorem~\ref{th:main}, but in light of the result of the theorem, the condition $x_\p(0)\in\XX_{\p}^0$ is a necessary requirement for~\eqref{eq:e(k)_asympt} and~\eqref{eq:y(k)_asympt} to hold.

We finally remark that the assumption in~\eqref{eq:e(k)_asympt} might be perceived as rather strong.
However, note that Assumption~\ref{ass:perfect-modeling} guarantees that model~\eqref{eq:model} can perfectly reproduce the input/output signals from the plant~\eqref{eq:plant} under perfect tracking conditions for a particular value of $\bar\theta$ and Assumption~\ref{ass:contr} guarantees that the controller can make the model track the reference for the same~$\bar\theta$.
Therefore, assumption~\eqref{eq:e(k)_asympt} amounts to having a well-chosen disturbance model and a well-designed state observer.
Note that this is a common assumption in the literature on offset-free MPC (see, e.g.,~\cite[Assumption~9]{pannocchia2015}).

\section{Offset-free EKF-based Nonlinear MPC} \label{sec:NMPC_EKF}

We now show how the result from the previous section may be applied to the frequently used combination of the Extended Kalman Filter (EKF) as observer and Nonlinear MPC (NMPC) as controller, where the EKF is used to learn the parameters of the disturbance model to reduce the tracking error and potentially lead to offset-free tracking under a suitable selection of the disturbance model.

\subsection{Observer: Extended Kalman Filter} \label{sec:EKF}

When considering the EKF, we take the following particularization of the prediction model~\eqref{eq:model}:
\begin{equation} \label{eq:model:dxdy}
	\begin{aligned}
		x(k+1) & =  f(x(k),u(k),d_x(k)), \\
        d_x(k) &= h_x(x(k),u(k), \theta(k)), \\
        d_y(k) &= h_y(x(k), \theta(k)), \\
		y(k) &=  g(x(k),d_y(k)),
	\end{aligned}
\end{equation}
where the disturbance $d$ is split into the process disturbance $d_x \in \R^\ndx$ and the output disturbance $d_y \in \R^\ndy$, and the disturbance function $h$ into $h_x \colon \R^\nx \times \R^\nu \times \R^\nth \to \R^\ndx$ and $h_y \colon \R^\nx \times \R^\nth \to \R^\ndy$.
Clearly, model~\eqref{eq:model} can be recovered from~\eqref{eq:model:dxdy} by taking $d = [d_x\T~ d_y\T]\T \in \R^\nd$ and $h$ similarly.
We note that $d$ is split so that the measurement-update of the EKF does not depend on the value of $u(k)$, as seen in the sequel.

Model~\eqref{eq:model:dxdy} can be interpreted as a combined model in which $f_n(x,u) \doteq f(x,u,0)$ and $g_n(x) \doteq g(x,0)$ capture a nominal model (either physics-based or black-box) estimated off-line from a set of input/output data $\{u(k),y(k)\}$, and the disturbance model is used for on-line adaptation to match the data measured from the real plant.

In particular,
one can estimate $(x(k), \theta(k))$ by using an EKF with measurement update
\begin{subequations}
\begin{equation} \label{eq:ekf-meas}
\begin{aligned} 
    d_y(k|k-1) & = h_y(x(k|k-1),\theta(k|k-1)), \\
    e(k) & =  y_\p(k) - g(x(k|k-1) , d_y(k|k-1)), \\
    B(k) &= C(k)P(k|k-1)C'(k) + Q_{y}(k), \\
    M(k) & = P(k|k-1)C'(k)B(k)^{-1}, \\
    \begin{bmatrix} x(k | k) \\ \theta(k | k) \end{bmatrix} & =
    \begin{bmatrix} x(k | k-1) \\ \theta( k | k-1) \end{bmatrix} + M(k)e(k), \\
    d_y(k|k)&=h_y(x(k | k),\theta(k | k)), \\
    P(k|k) & = \big(I - M(k)C(k)\big)P(k|k-1),
\end{aligned}
\end{equation}
and time update
\begin{equation}
\begin{aligned} \label{eq:ekf-time}
    d_x(k|k)&=h_x(x(k | k),u(k),\theta(k | k)), \\
    x(k+1 | k) & =  f( x(k|k), u(k), d_x(k|k) ), \\
    \theta(k+1 | k) & = \theta(k|k), \\
    P(k+1 | k) & = A(k)P(k|k)A' (k) + Q(k), 
\end{aligned}
\end{equation}
where
\begin{equation}
\begin{aligned} \label{eq:ekf_matrices}
    \moveEq{-10}C(k) & {=} \begin{bmatrix} \left(\frac{\partial g}{\partial x} {+}
        \frac{\partial g}{\partial d_y}\frac{\partial h_y}{\partial x}\right)
        & \frac{\partial g}{\partial d_y}\frac{\partial h_y}{\partial \theta}
    \end{bmatrix}\Big|_{x(k | k-1),\theta(k | k-1)},  \\
            \moveEq{-10}A(k) & {=} \begin{bmatrix} \left(\frac{\partial f}{\partial x} {+}
        \frac{\partial f}{\partial d_x}\frac{\partial h_x}{\partial x}\right)
        & \frac{\partial f}{\partial d_x} \frac{\partial h_x}{\partial \theta} \\ 0 & I
    \end{bmatrix}\Bigg|_{x(k | k),\theta(k | k),u(k)},  \\
                    \moveEq{-10}Q(k) & {=} \begin{bmatrix} Q_x(k) & 0 \\ 0 & Q_{\theta}(k) \end{bmatrix},
\end{aligned}%
\end{equation}
\label{eq:ekf}%
\end{subequations}
and $Q_x(k)$, $Q_\theta(k)$, and $Q_y(k)$ are positive semidefinite matrices representing, respectively, the covariance matrices of the process noise, parameter noise and output noise.\footnote{We note that the EKF~\eqref{eq:ekf} implicitly assumes that model~\eqref{eq:model:dxdy} is at least differentiable in its arguments.}

\subsection{Controller: Nonlinear MPC} \label{sec:controller}
Consider the NMPC formulation
\begin{equation}
	\begin{aligned}
        \min_{x,u}  & \sum_{j=0}^{N-1} \ell(x_j,u_j,x_\r(k{+}j),u_\r(k{+}j))\\[-1em]
        &\hspace*{1cm}+ V_\mathrm{f}(x_N,x_\r(k{+}N),\theta(k|k))\\
		\mathrm{s.t.} \ & x_0 = x(k|k), \\
                        &x_{j+1} = f(x_j,u_j,d_j), \; j \in \N_0^{N-1}, \\
                        &d_j = h(x_j,u_j, \theta(k|k)), \; j \in \N_0^{N-1}, \\
                        &y_j = g(x_j,d_j), \; j \in \N_0^{N-1}, \\
                        &y_j \in \YY,\; u_j \in \UU, \; j \in \N_0^{N-1}, \\
                        &x_N \in \mathcal{X}_\mathrm{f}(x_\r(k+N),\theta(k|k)),
	\end{aligned}
	\label{eq:nmpc}%
\end{equation}
where $\ell$ is the stage cost, $V_\mathrm{f}$ the terminal cost, and $\mathcal{X}_\mathrm{f}$ the terminal set. The corresponding control law is $u(k)=u_0^*$, where $u_0^*$ is the first control move obtained from the optimal solution $(x^*,u^*,d^*,y^*)$ of~\eqref{eq:nmpc}.
Design conditions for $\ell$, $V_\mathrm{f}$ and $\mathcal{X}_\mathrm{f}$ under which Assumption~\ref{ass:contr} is satisfied are well known for the time-invariant case, i.e., when $\theta(k|k)$ is constant at all time instants $k$, see~\cite{RMD17,limon2009input}.
However, due to the time-varying nature of $\theta(k|k)$ and $r(k)$, additional care should be taken when designing $V_\mathrm{f}$ and $\mathcal{X}_\mathrm{f}$ so as to guarantee stability and recursive feasibility, as well as to provide robustness under vanishing perturbations.
Due to space considerations, we leave these technical details out of the scope of this paper, instead relying on taking the classical stage cost 
\begin{align*}
    \ell(x, u, x_r, u_r) &= \| x - x_r \|^2_{W_x} + \| u - u_r \|^2_{W_u},
\end{align*}
and a terminal equality constraint, i.e., $\mathcal{X}_\mathrm{f} = \{x_r(k+N)\}$, with $V_\mathrm{f} = 0$.
This NMPC setting, in general, does not guarantee \emph{a priori} closed-loop stability and recursive feasibility properties when the prediction model changes.
However, if the mismatch between the process~\eqref{eq:plant} and the nominal model~\eqref{eq:model} is small and $\theta(k|k)$ changes slowly, such properties may occur in practice.

At each sample time $k$,~\eqref{eq:nmpc} requires reference signals $x_\r(k+j)$, $j \in \N_0^{N}$, and $u_\r(k+j)$, $j \in \N_0^{N-1}$, that satisfy the prediction model for the current estimate of the disturbance parameters $\theta(k|k)$.
In order to compute these signals for the output reference trajectory $\{r(k+j)\}_{j=0}^\infty$, consider the following infinite-horizon reference optimization problem
\begin{align} \label{eq:ref_signals}
    \min\limits_{\vv{\hat{x}}_\r,\vv{\hat{u}}_\r} 
    &\;\sum_{j=0}^\infty \ell_\r(\hat{x}_\r(k+j), \hat{u}_\r(k+j)) \\ 
    \st &\; r(k+j)=g(\hat{x}_\r(k+j),\hat{d}_\r(k+j)),\nonumber \\
        &\; \hat{x}_\r(k+j+1)=f(\hat{x}_\r(k+j),\hat{u}_\r(k+j),\hat{d}_\r(k+j)), \nonumber\\
        &\; \hat{d}_\r(k+j)=h(\hat{x}_\r(k+j),\hat{u}_\r(k+j),\theta(k|k)), \nonumber \\
        &\; \hat{u}_\r(k+j) \in \UU, \nonumber
\end{align}
where $\vv{\hat{x}}_\r=\{\hat{x}_\r(k+j)\}_{j=0}^\infty$ and
$\vv{\hat{u}}_\r=\{\hat{u}_\r(k+j)\}_{j=0}^\infty$ are, respectively, the state and input reference sequences associated with the reference trajectory $\{r(k+j)\}_{j=0}^\infty$, and the cost function $\ell_\r\colon\R^\nx\times\R^\nu\to\R$ is any convex function that can be used to make the selection unique in case of multiple solutions.
A typical choice is to take
\begin{equation} \label{eq:ref_signals:cost}
    \ell_r(\hat{x}(k+j), \hat{u}(k+j)) = \| \hat{u}(k+j) - u_d(k+j) \|^2,
\end{equation}
where $\{u_d(k)\}_{k=1}^\infty$ is a sequence describing a \emph{desired} input trajectory of the system, which is often available since inputs are typically related to aspects such as energy consumption.
An ideal choice is to take $u_d(k) = u_\r(k)$, where $u_\r(k)$ is the reference signal given by Assumption~\ref{ass:trackability}, if available.

At each sample time $k$, the reference signals $x_\r(k+j)$ and $u_\r(k+j)$ of \eqref{eq:nmpc} are taken from the optimal solution of \eqref{eq:ref_signals}.
In a practical setting, to avoid having to solve an infinite-horizon optimization problem, we can modify~\eqref{eq:ref_signals} to only consider $M \in \N$ future samples of the reference signal, i.e., to consider problem~\eqref{eq:ref_signals} for $j \in \N_0^M$, instead of $j \in \N_0^\infty$, where $M \geq N$, since \eqref{eq:nmpc} requires $N$ future samples of the reference.
A special case of~\eqref{eq:ref_signals} is a constant set-point $r(k+j)\equiv r(k)$, $\forall j \in \N$, for which a feasible solution $\vv{\hat{x}}_\r,\vv{\hat{u}}_\r$ is given by solving the steady-state equations
\begin{equation} \label{eq:ref_signals-no-preview}
	\begin{aligned}
        r(k)&=g(\hat{x}_\r(k),\hat{d}_\r(k)), \\
        \hat{x}_\r(k)&=f(\hat{x}_\r(k),\hat{u}_\r(k),\hat{d}_\r(k)), \\ 
        \hat{d}_\r(k)&=h(\hat{x}_\r(k),\hat{u}_\r(k),\theta(k|k)).
	\end{aligned}
\end{equation}
Another interesting instance of~\eqref{eq:ref_signals} is the special case
of periodic reference signals with period $T$, i.e., reference signals satisfying $r(k+j+T)=r(k+j)$, $\forall j \in \N$.
In this case, problem~\eqref{eq:ref_signals} would consider $j \in \N_0^{T-1}$, instead of $j \in \N_0^\infty$, and include the additional constraint $\hat{x}_\r(k) = \hat{x}_\r(k+T)$.

\begin{remark} \label{rem:NMPCT}
    In a practical setting, it may not be possible to choose a $r(k)$ that is known to satisfy Assumptions~\ref{ass:trackability} and~\ref{ass:stabilizability} due to a lack of knowledge of the real plant dynamics~\eqref{eq:plant}.
    In this case, the use of nonlinear MPC formulations with artificial reference \cite{limon_NMPCT_2018, kohler_AUT_2020}, could be used to provide convergence to the closest admissible reference trajectory to the given $r(k)$.
\end{remark}

\vspace*{1em}
\subsection{Numerical examples} \label{sec:results}

\begin{figure*}[!h]
        \centering
        \subfloat[Polynomial disturbance model]{
            \begin{minipage}{0.33\linewidth}
                \includegraphics[width=0.98\linewidth, height = 0.2\textheight, keepaspectratio=true]{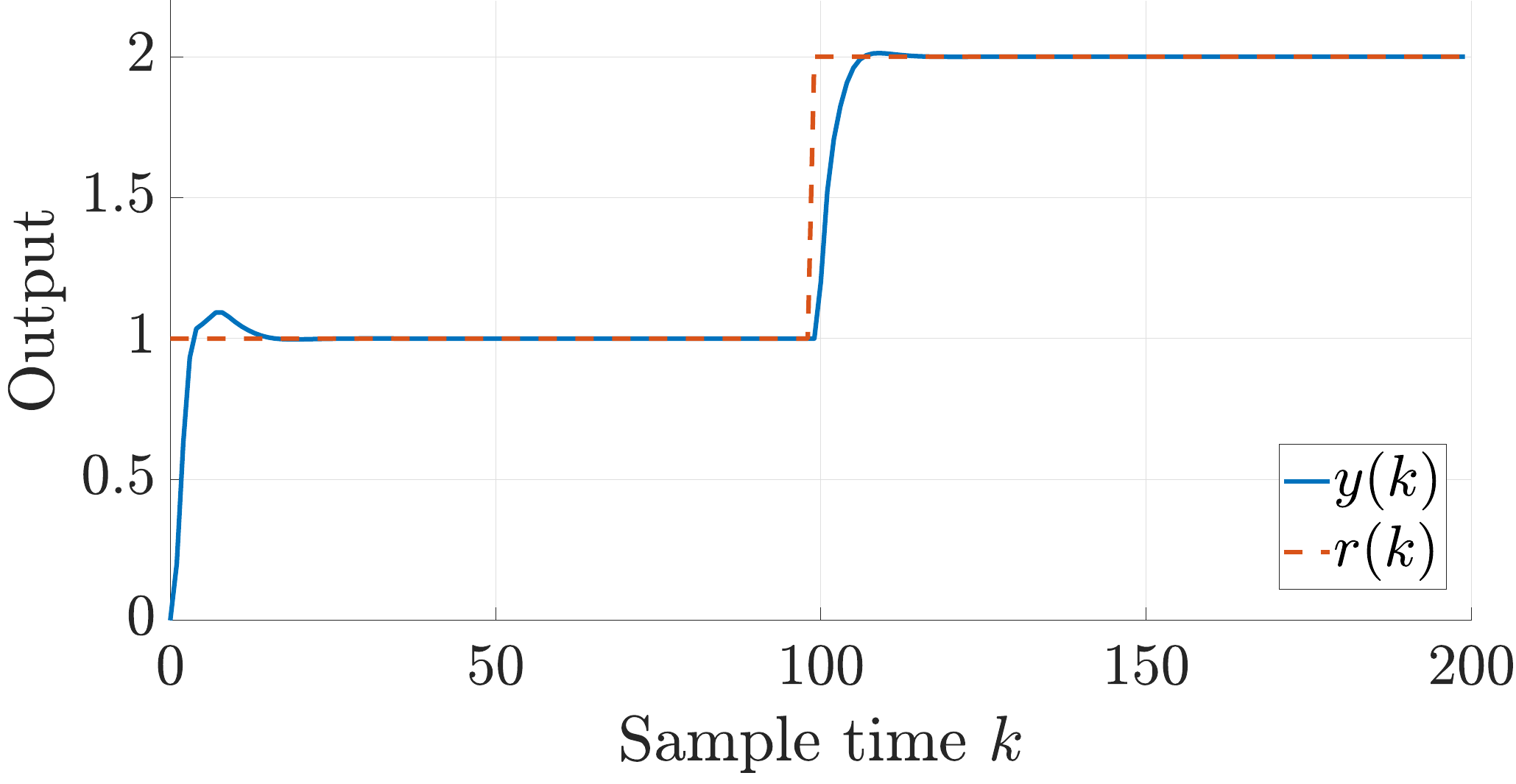}

                \includegraphics[width=0.98\linewidth, height = 0.2\textheight, keepaspectratio=true]{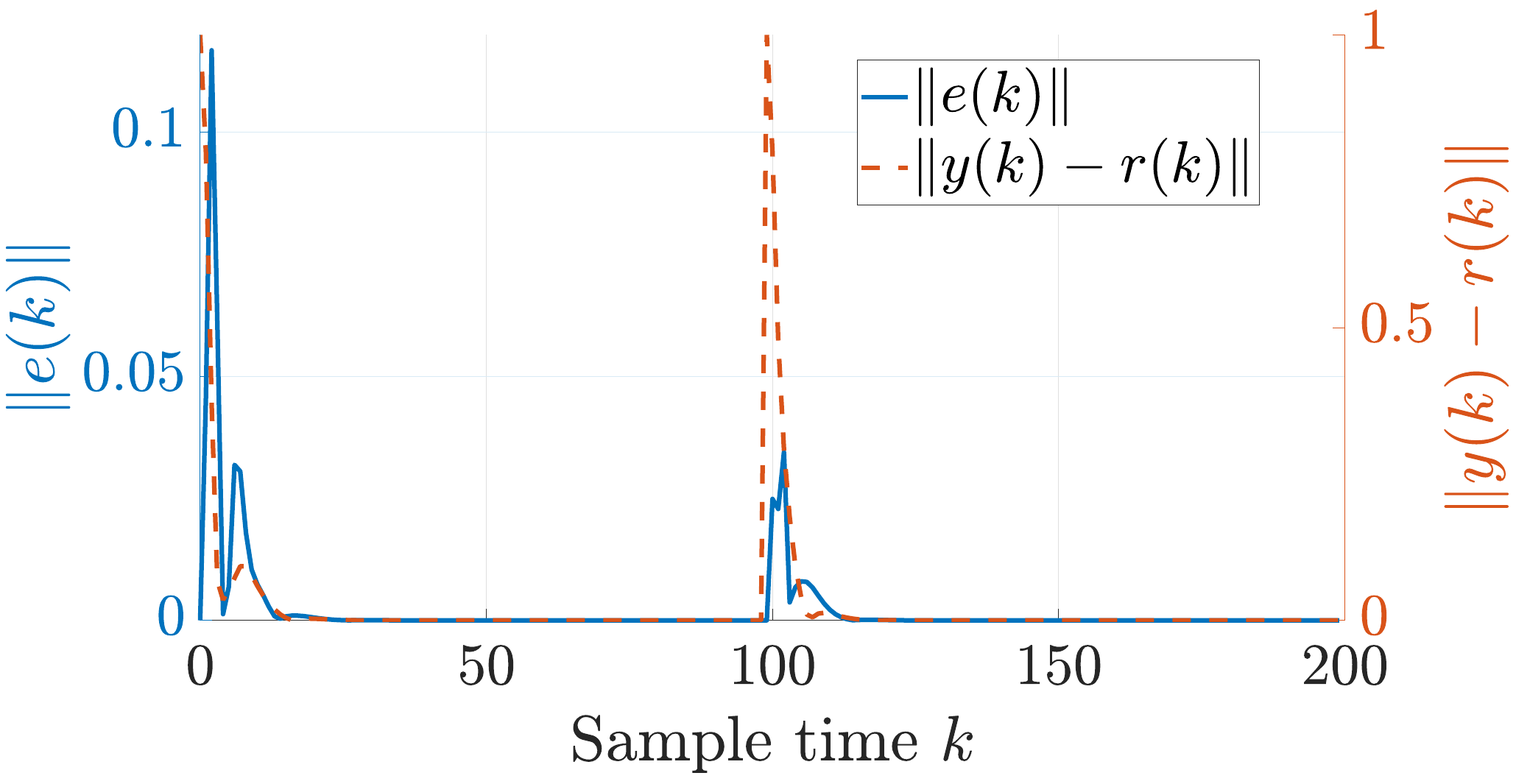}
            \end{minipage}} 
        \subfloat[Constant disturbance model]{
            \begin{minipage}{0.33\linewidth}
                \includegraphics[width=0.98\linewidth, height = 0.2\textheight, keepaspectratio=true]{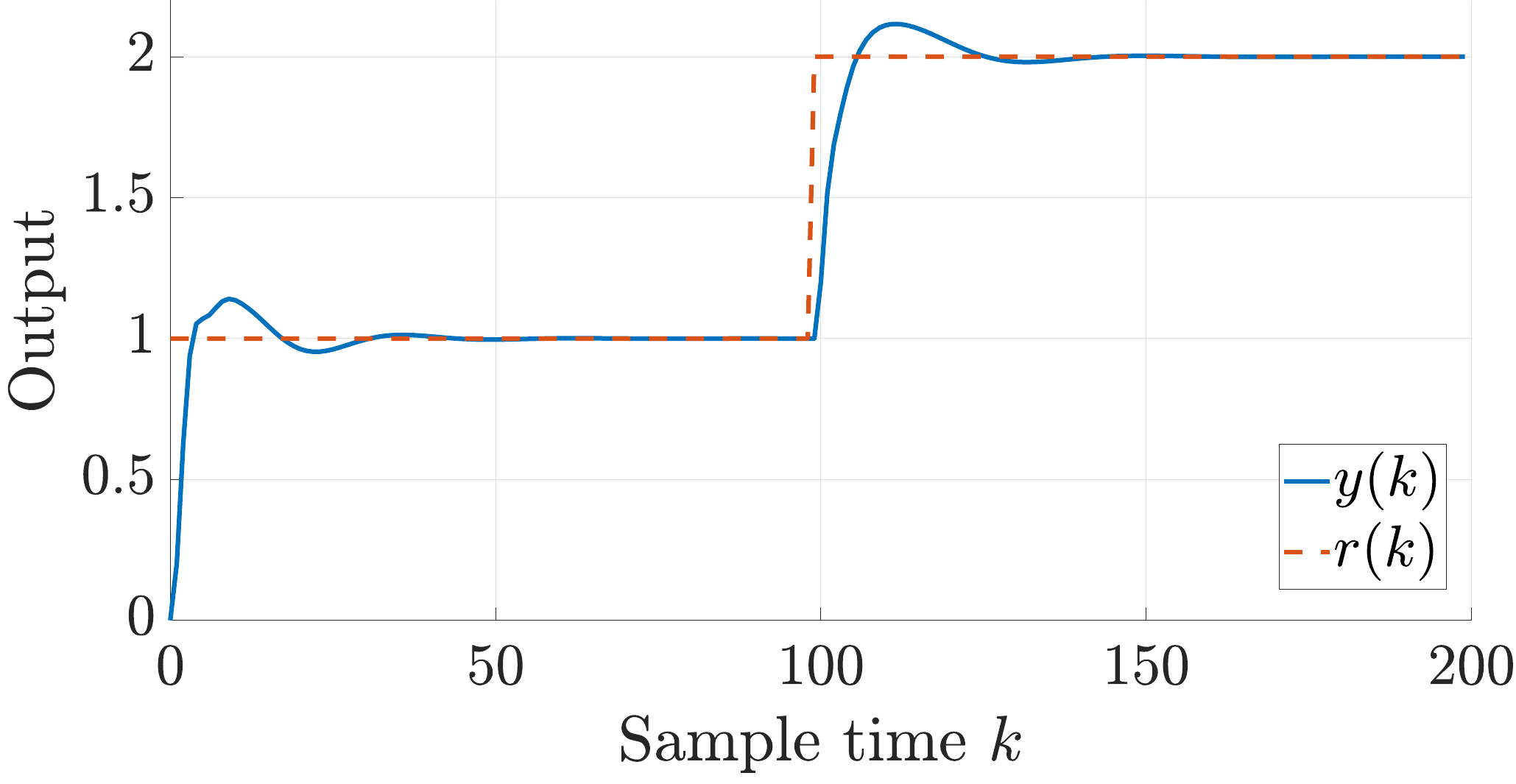}

                \includegraphics[width=0.98\linewidth, height = 0.2\textheight, keepaspectratio=true]{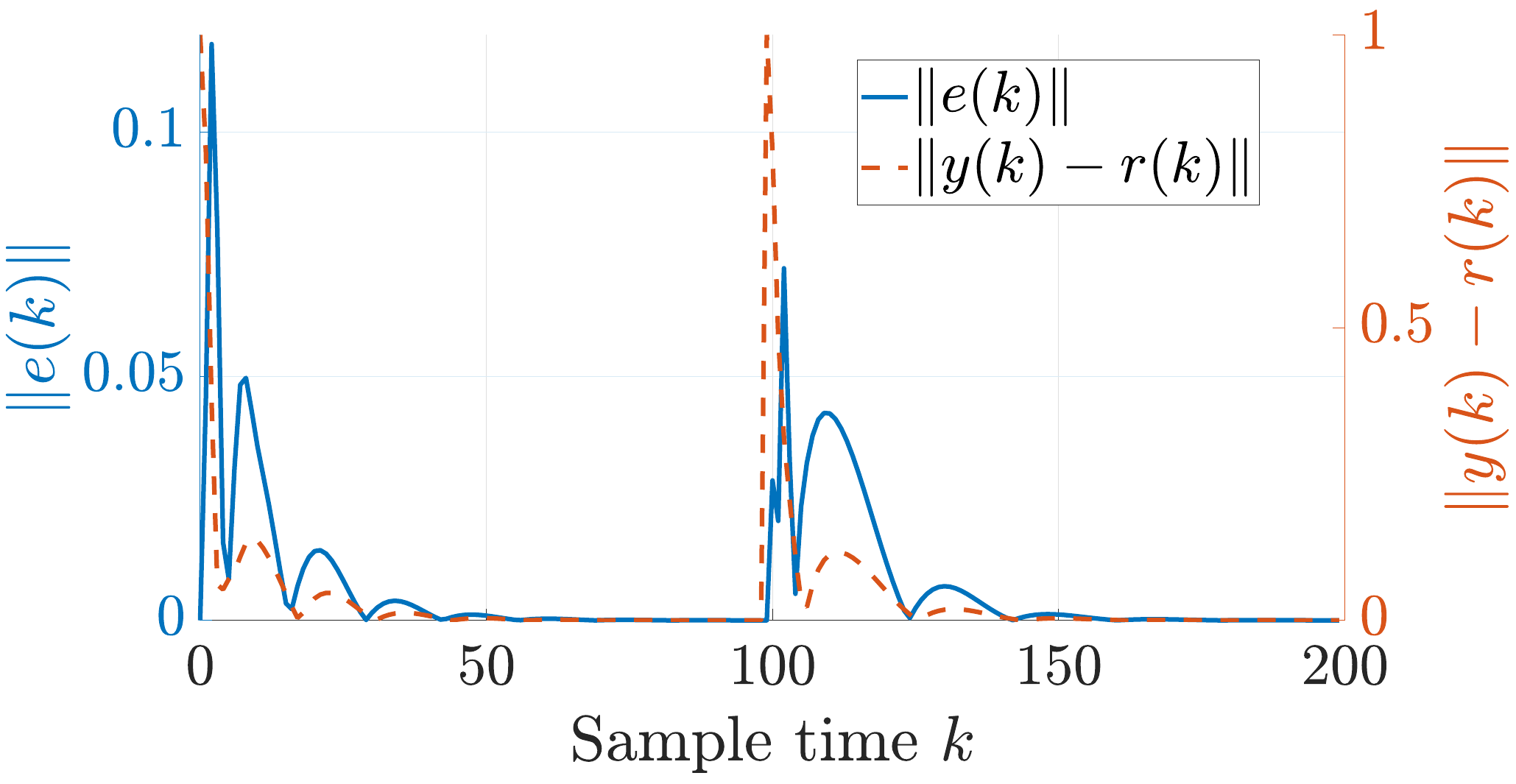}
            \end{minipage}} 
        \subfloat[FNN disturbance model]{
            \begin{minipage}{0.33\linewidth}
                \includegraphics[width=0.98\linewidth, height = 0.2\textheight, keepaspectratio=true]{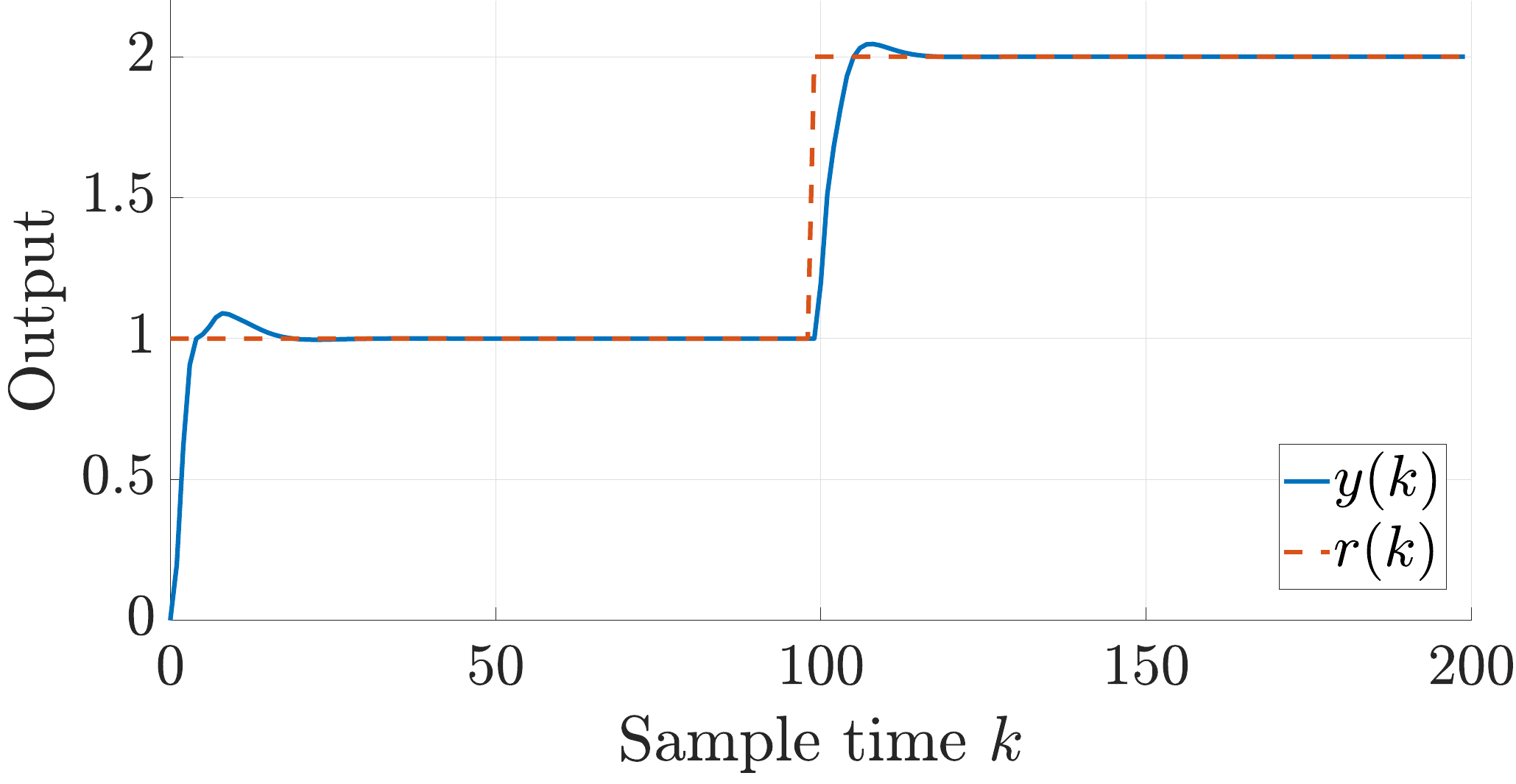}

                \includegraphics[width=0.98\linewidth, height = 0.2\textheight, keepaspectratio=true]{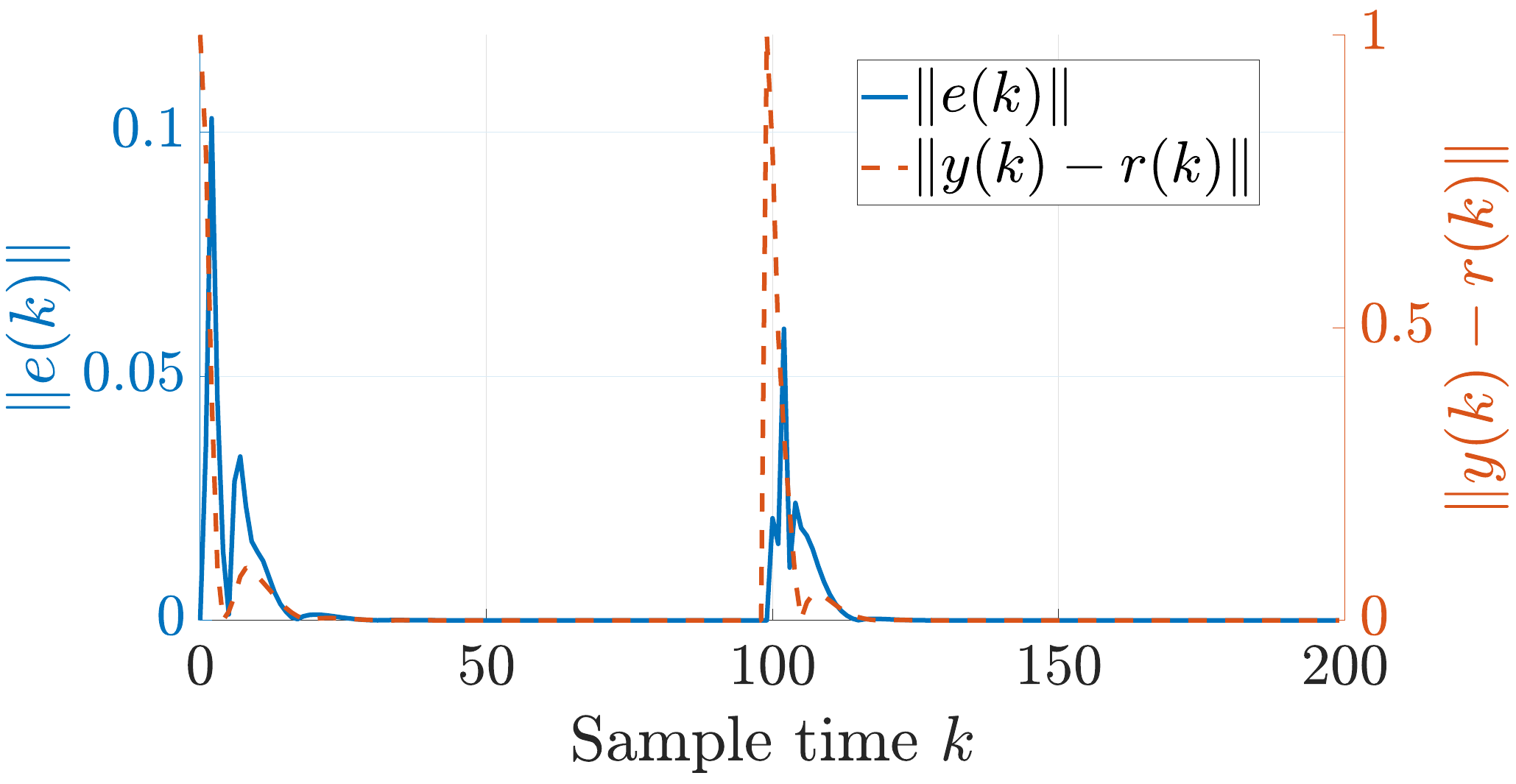}
            \end{minipage}} 
        \caption{Van der Pol plant. Closed-loop tracking of a piecewise-constant reference using different disturbance models.}
        \label{fig:VanderPol_ss}
\end{figure*}

\begin{figure*}[!h]
        \centering
        \subfloat[Polynomial disturbance model]{\label{fig:VanderPol_gen:poly}
            \begin{minipage}{0.33\linewidth}
                \includegraphics[width=0.98\linewidth, height = 0.2\textheight, keepaspectratio=true]{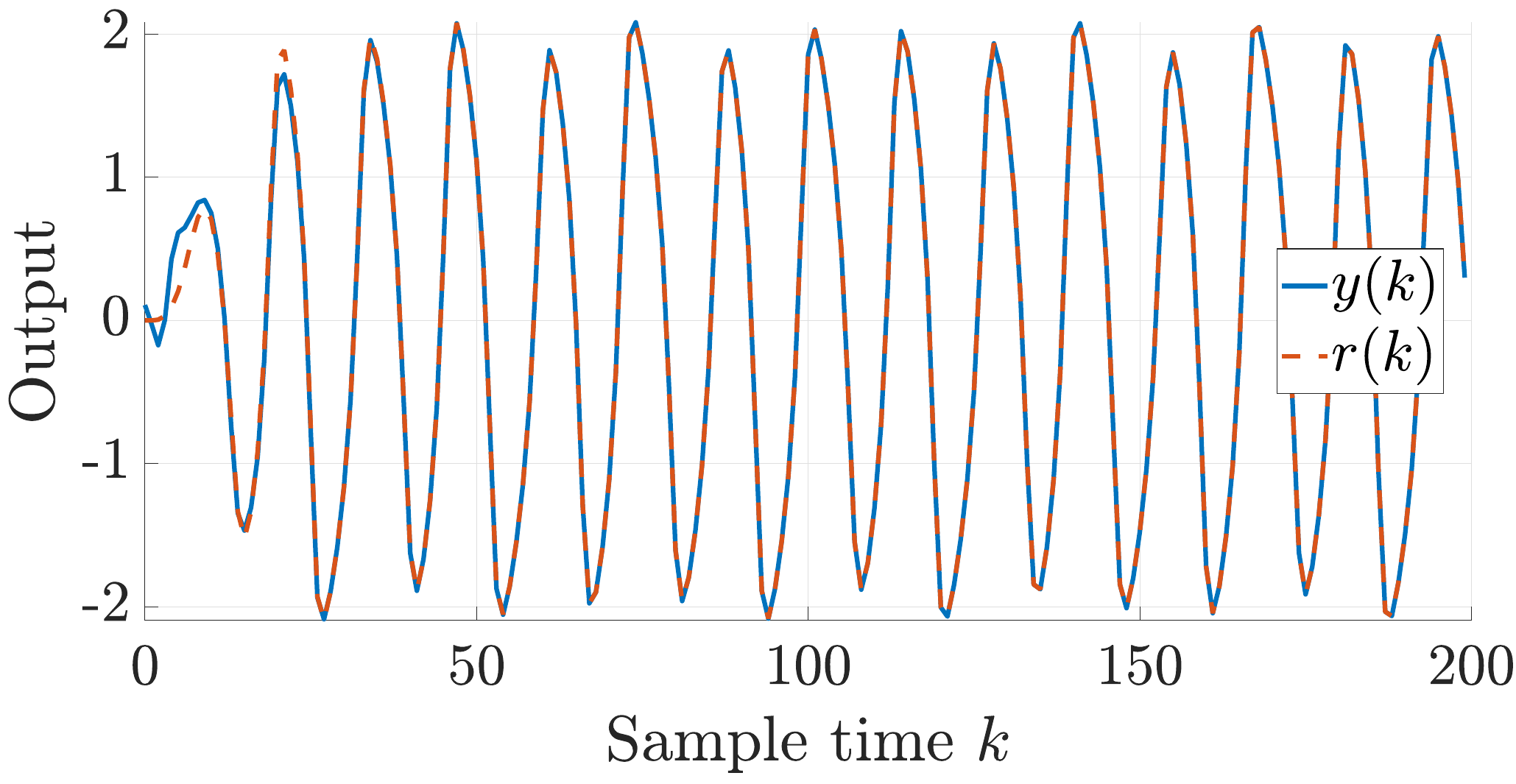}

                \includegraphics[width=0.98\linewidth, height = 0.2\textheight, keepaspectratio=true]{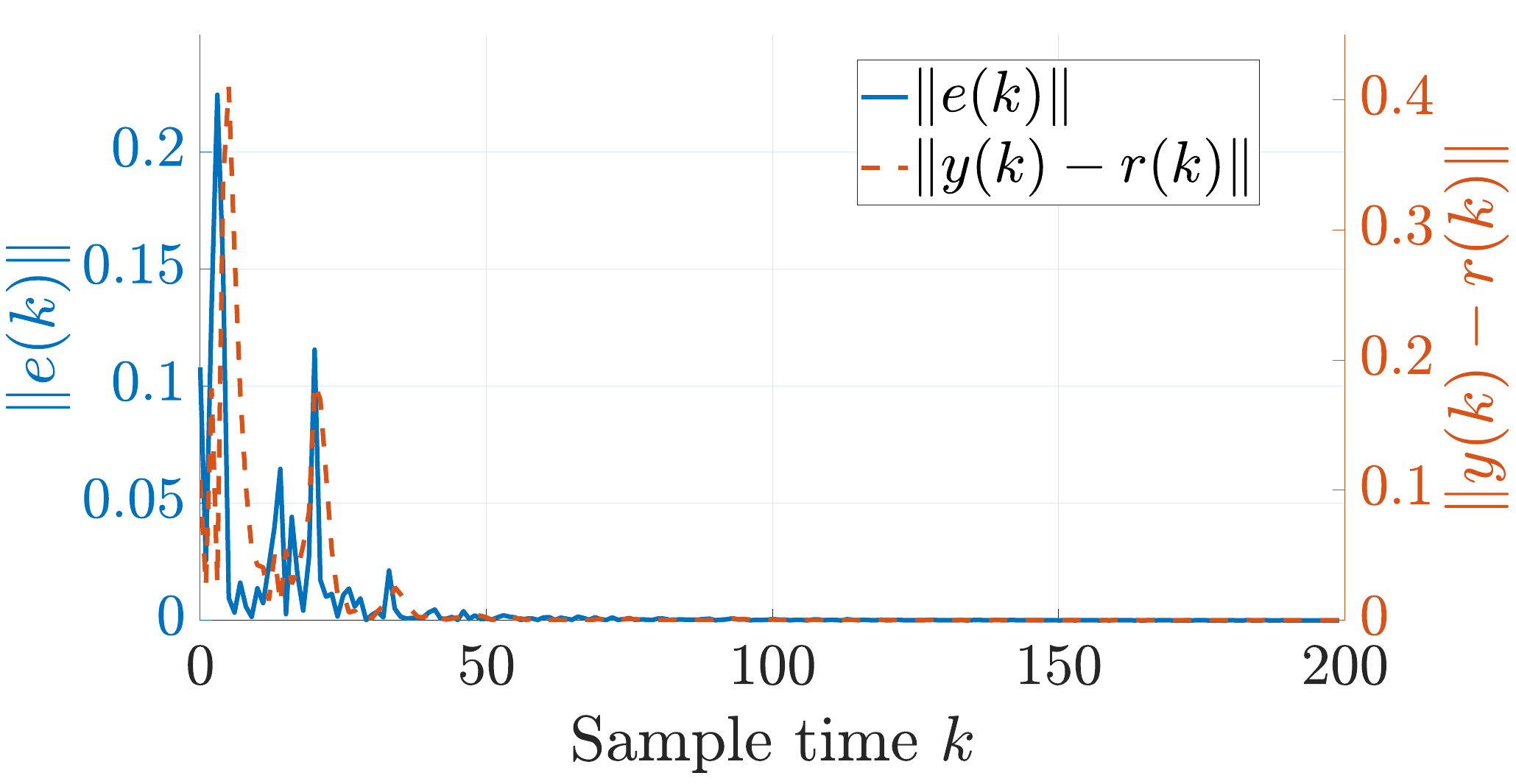}
            \end{minipage}} 
        \subfloat[Constant disturbance model]{\label{fig:VanderPol_gen:const}
            \begin{minipage}{0.33\linewidth}
                \includegraphics[width=0.98\linewidth, height = 0.2\textheight, keepaspectratio=true]{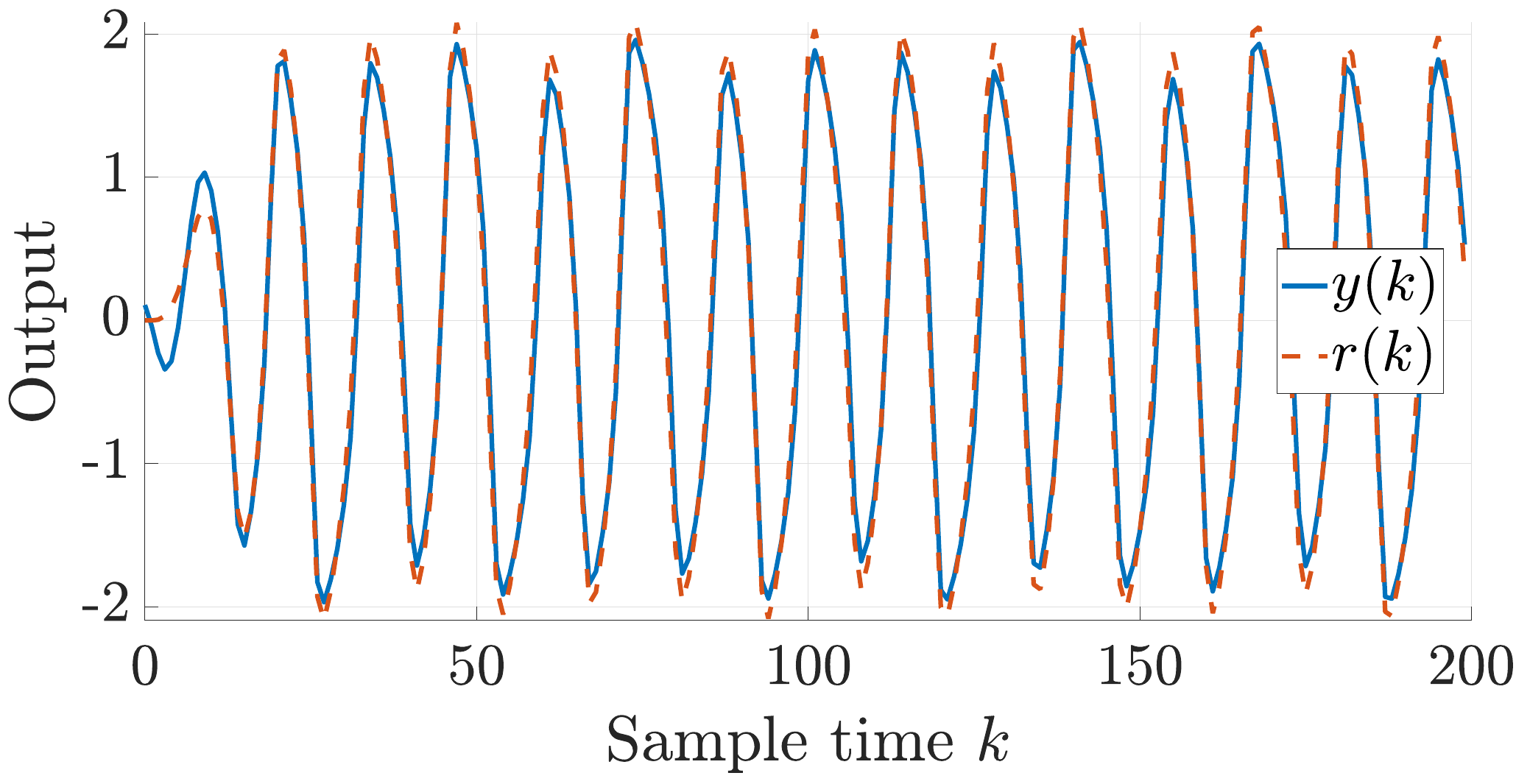}

                \includegraphics[width=0.98\linewidth, height = 0.2\textheight, keepaspectratio=true]{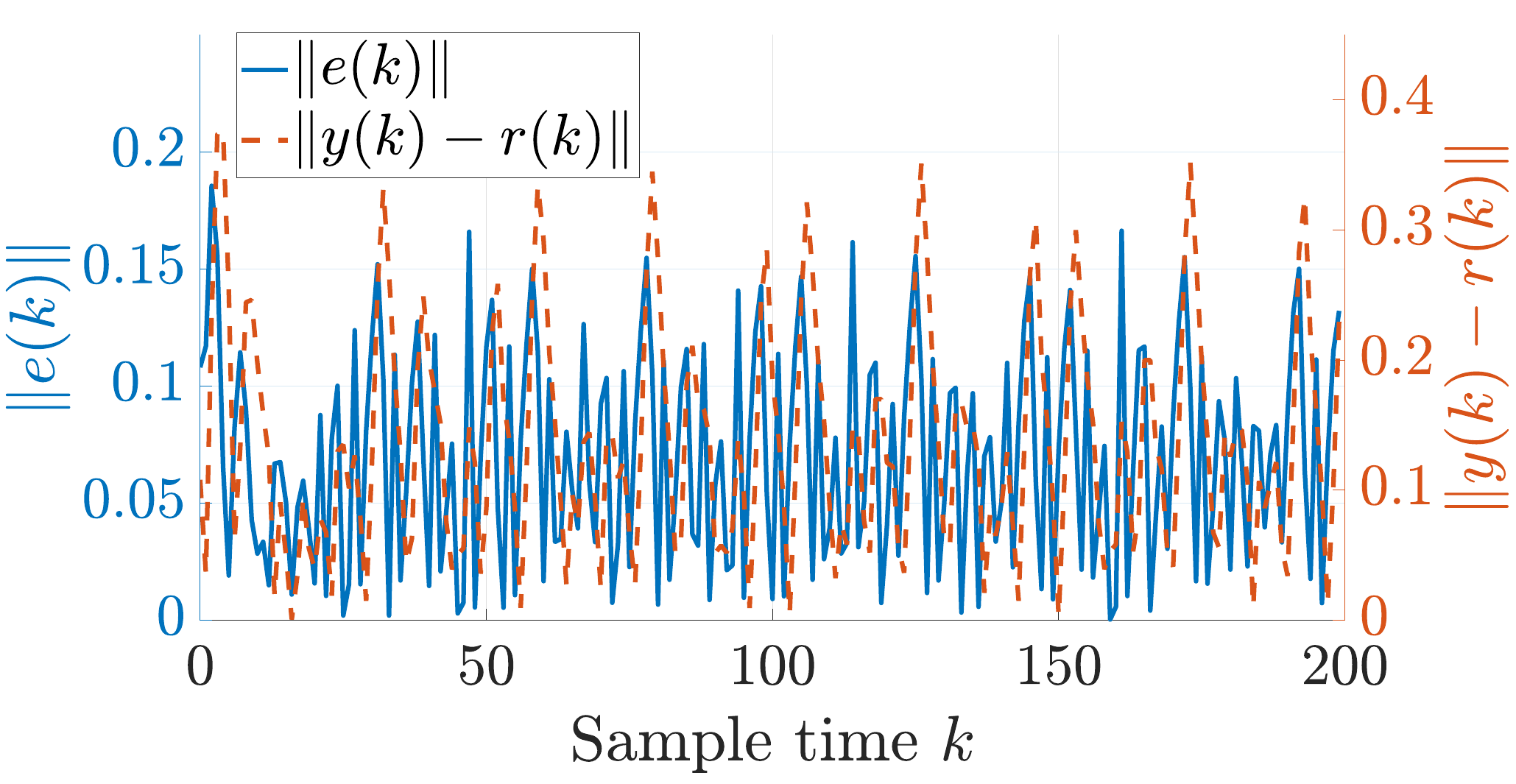}
            \end{minipage}} 
        \subfloat[FNN disturbance model]{ \label{fig:VanderPol_gen:FNN}
            \begin{minipage}{0.33\linewidth}
                \includegraphics[width=0.98\linewidth, height = 0.2\textheight, keepaspectratio=true]{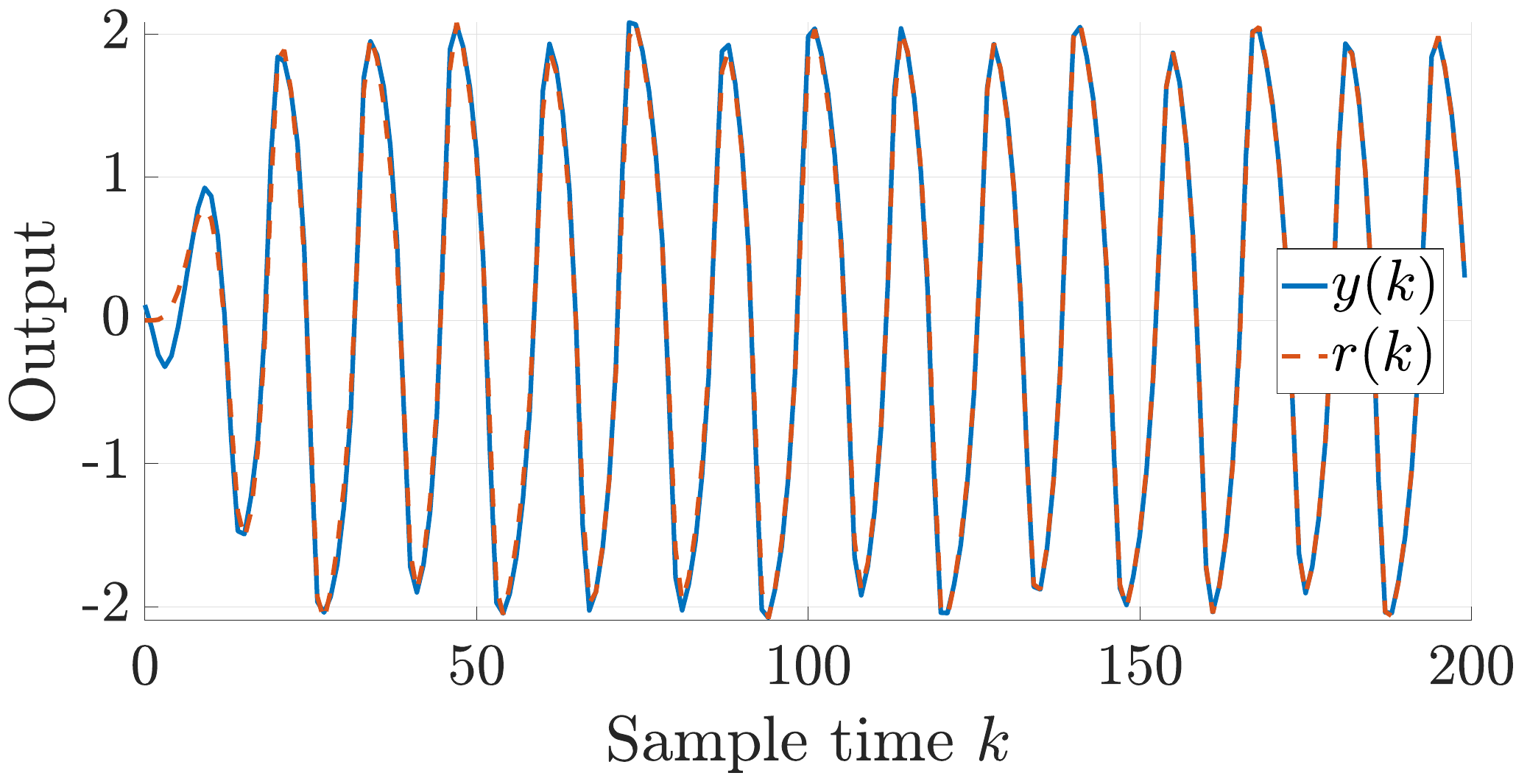}

                \includegraphics[width=0.98\linewidth, height = 0.2\textheight, keepaspectratio=true]{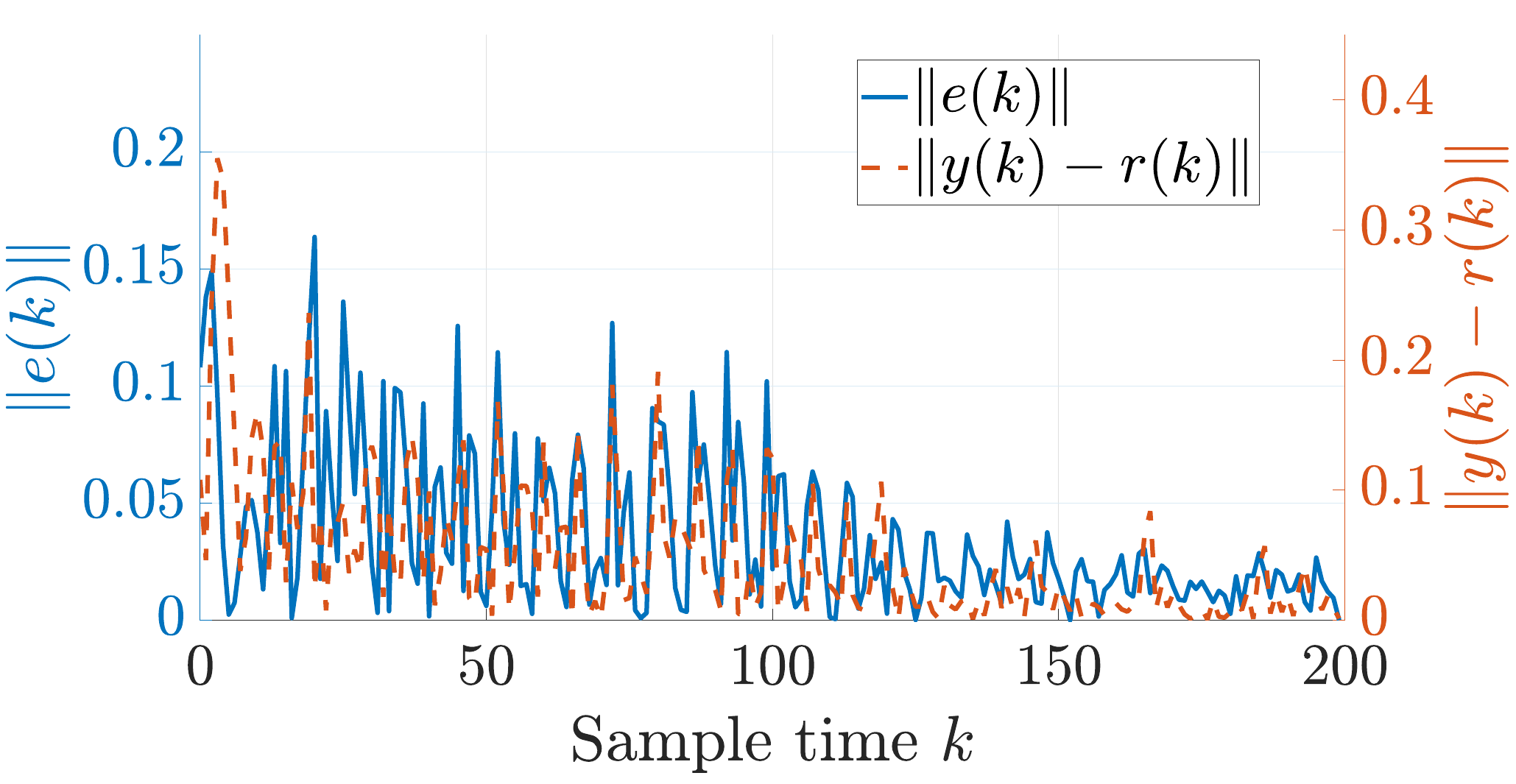}
            \end{minipage}} 
        \caption{Van der Pol plant. Closed-loop tracking of a generic reference trajectory using different disturbance models.}
        \label{fig:VanderPol_gen}
\end{figure*}

\begin{figure*}[!h]
        \centering
        \subfloat[Polynomial disturbance model]{\label{fig:CSTR_gen:poly}
            \begin{minipage}{0.33\linewidth}
                \includegraphics[width=0.98\linewidth, height = 0.2\textheight, keepaspectratio=true]{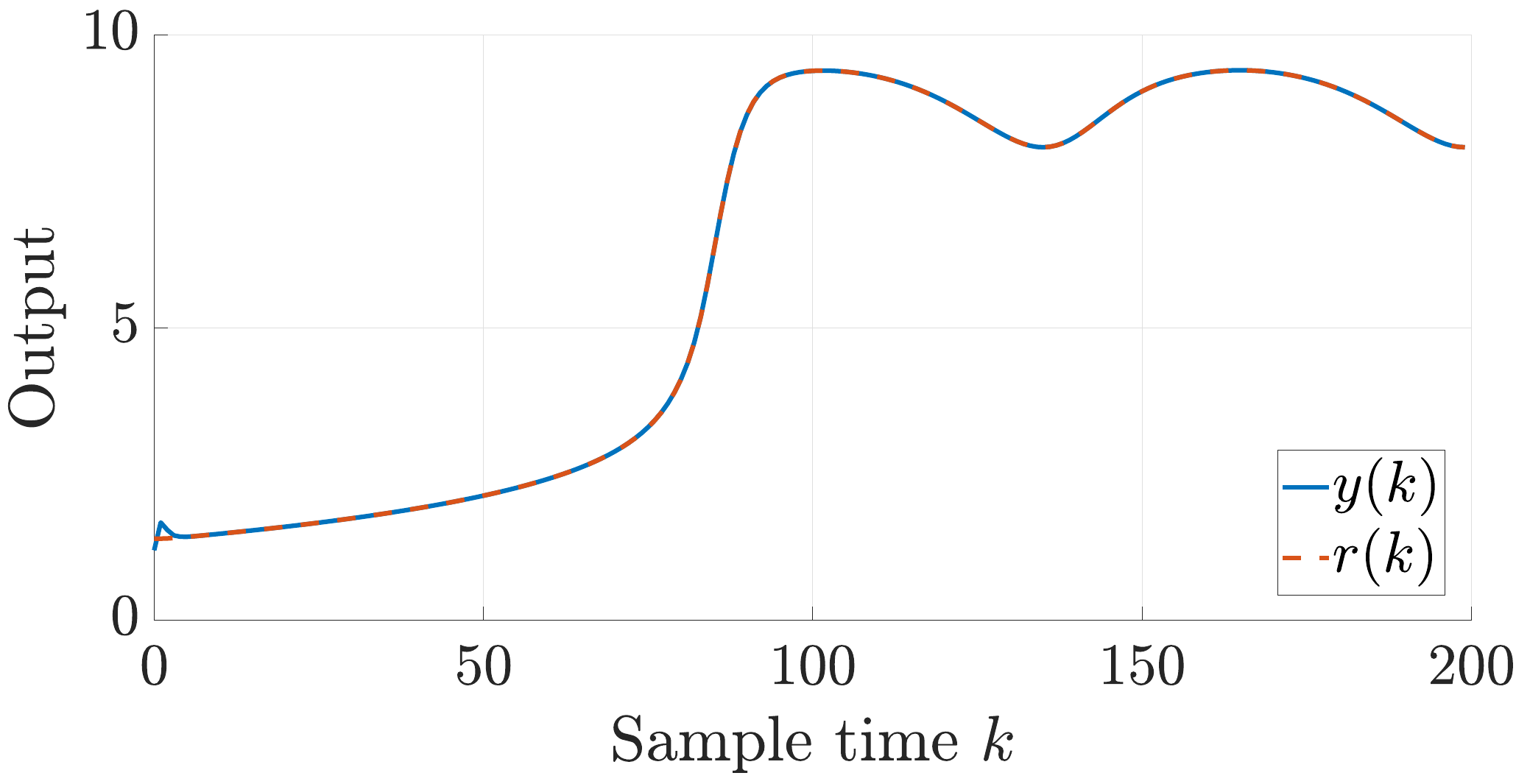}

                \includegraphics[width=0.98\linewidth, height = 0.2\textheight, keepaspectratio=true]{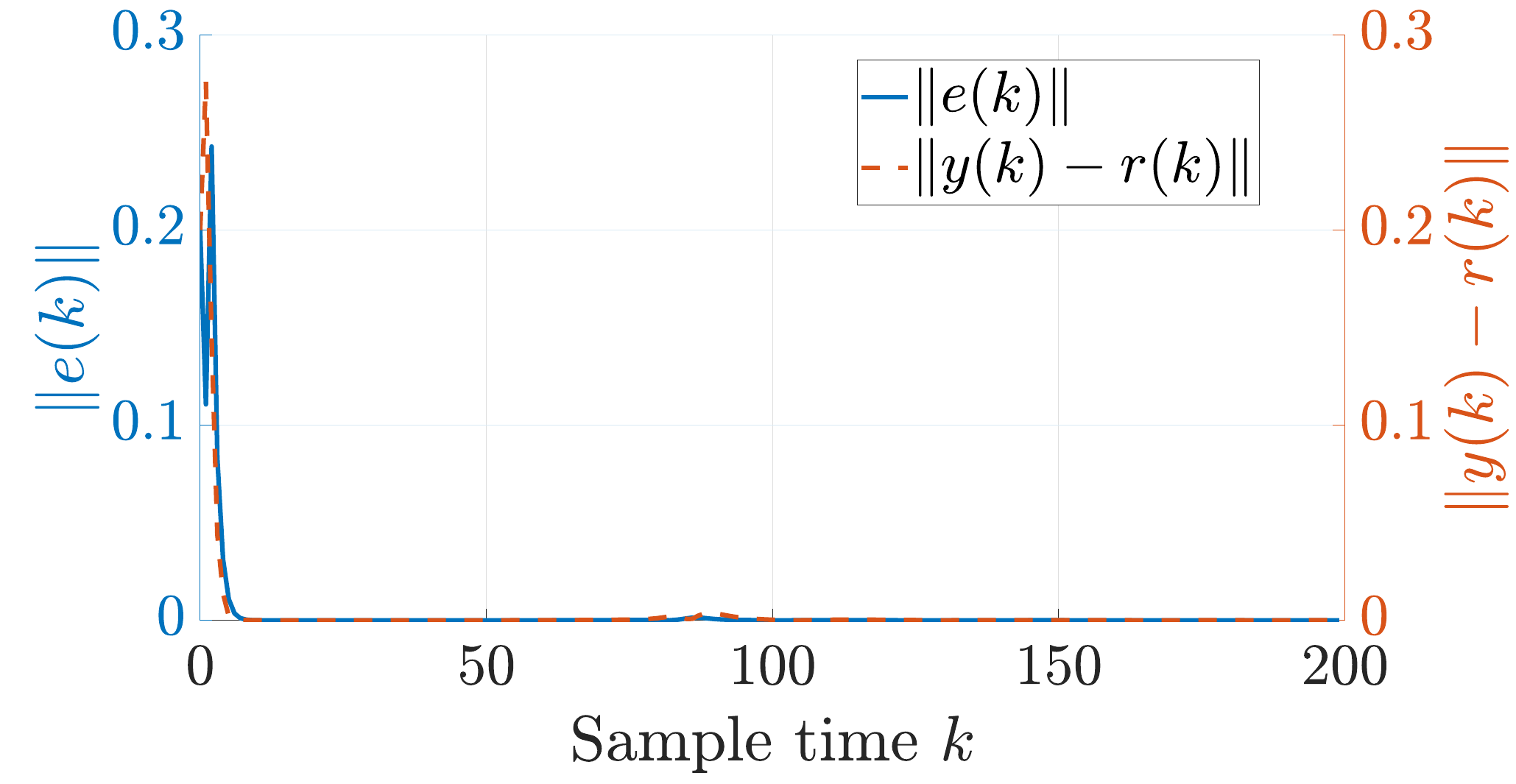}
            \end{minipage}} 
        \subfloat[Constant disturbance model]{\label{fig:CSTR_gen:const}
            \begin{minipage}{0.33\linewidth}
                \includegraphics[width=0.98\linewidth, height = 0.2\textheight, keepaspectratio=true]{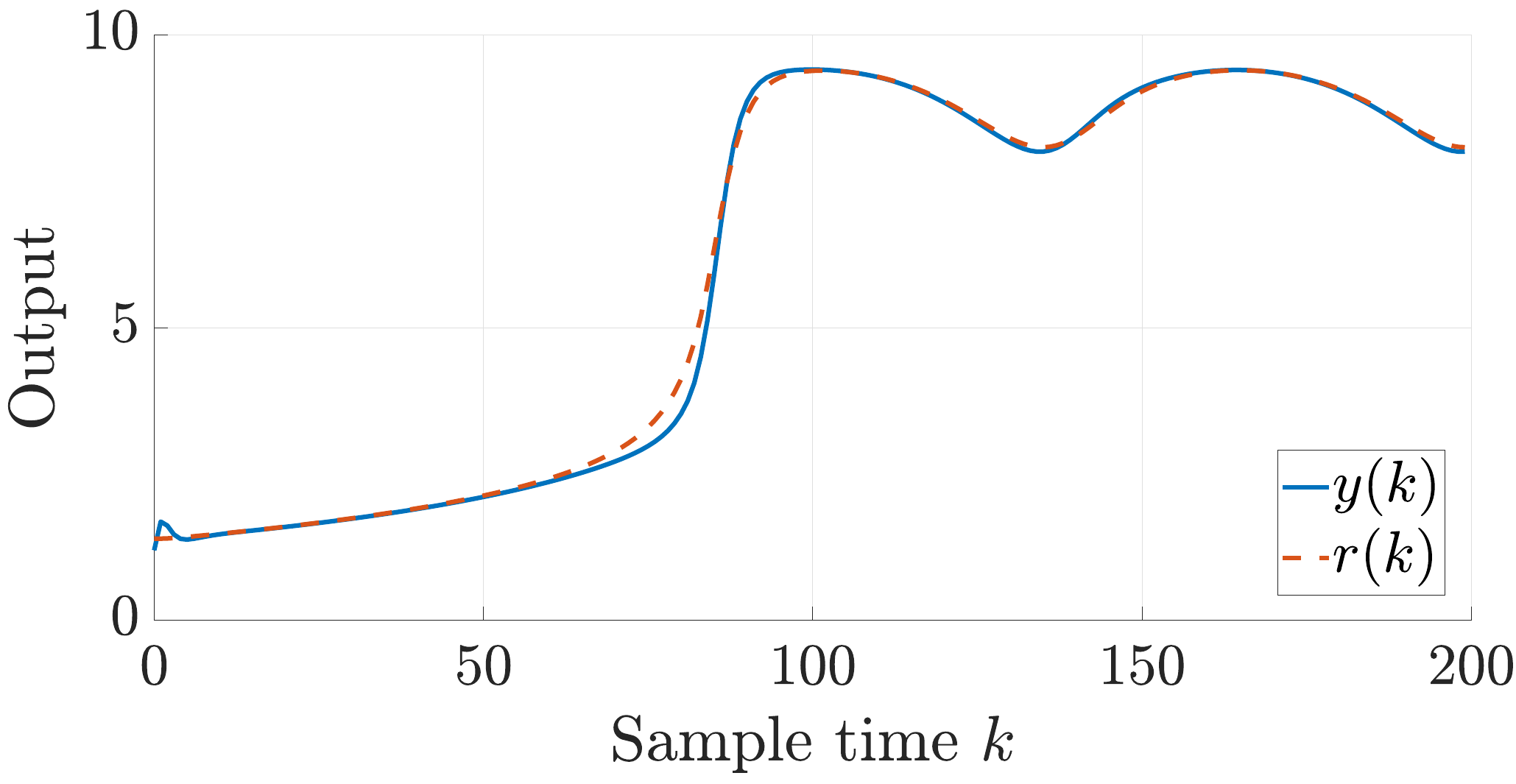}

                \includegraphics[width=0.98\linewidth, height = 0.2\textheight, keepaspectratio=true]{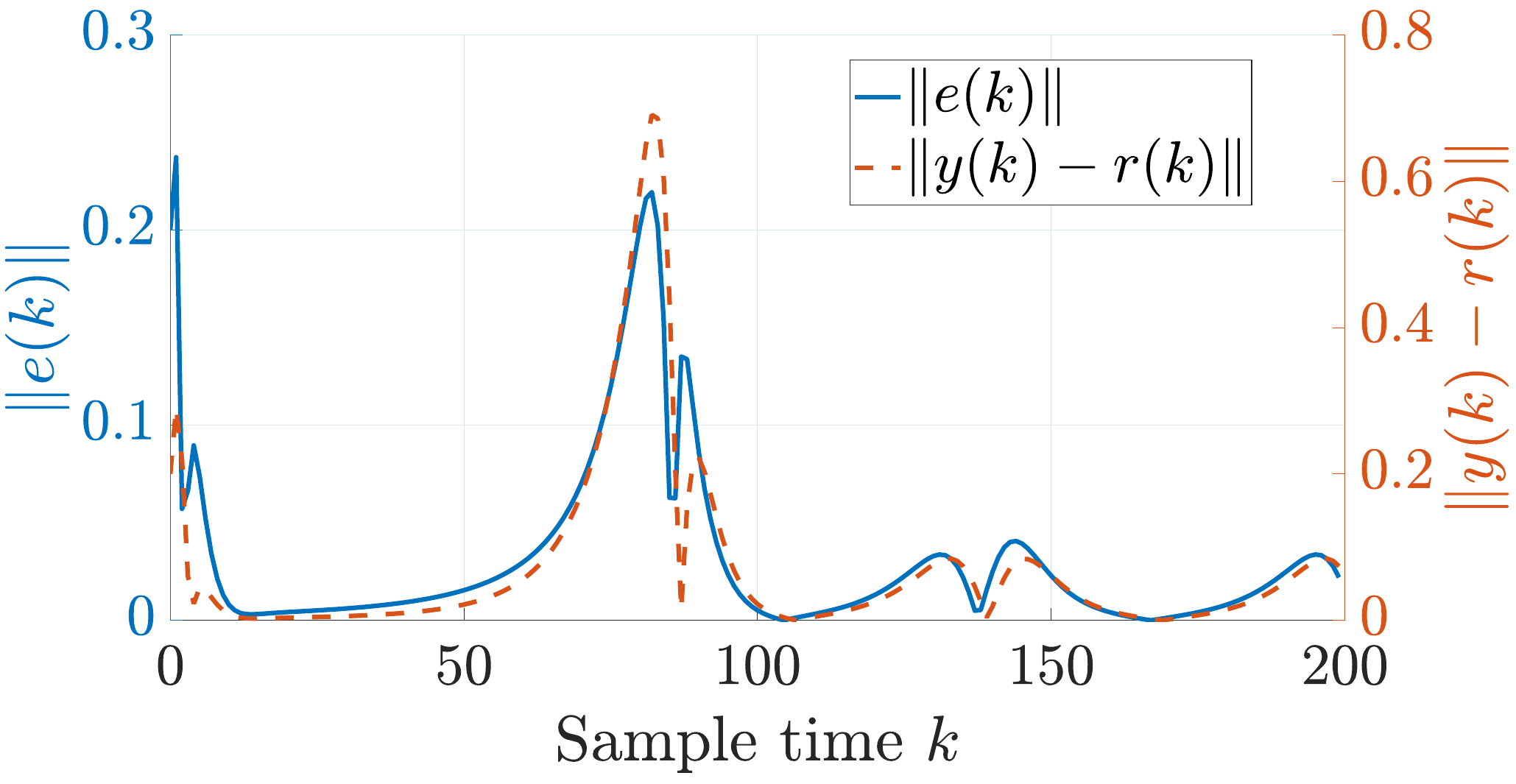}
            \end{minipage}} 
        \subfloat[FNN disturbance model]{ \label{fig:CSTR_gen:FNN}
            \begin{minipage}{0.33\linewidth}
                \includegraphics[width=0.98\linewidth, height = 0.2\textheight, keepaspectratio=true]{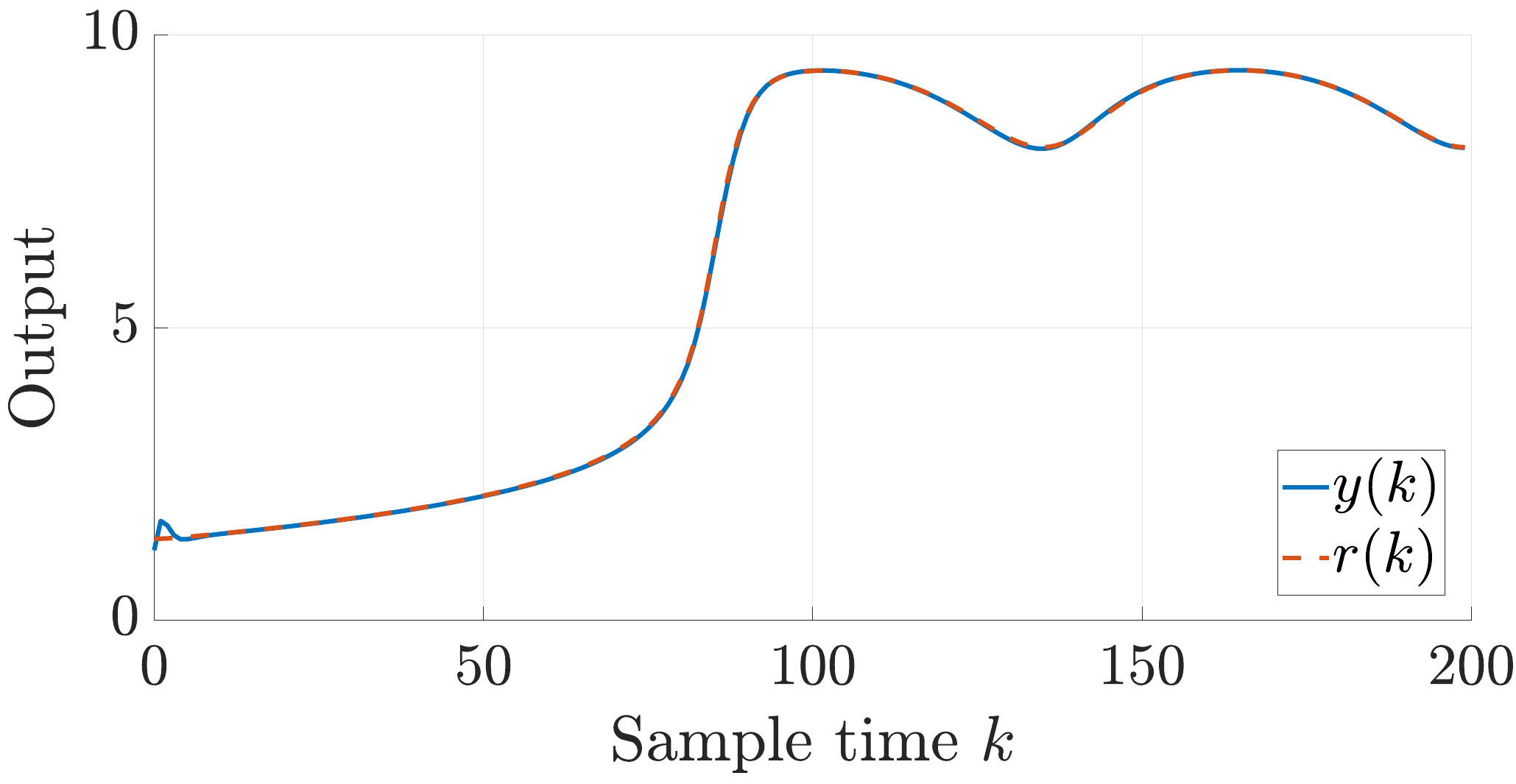}

                \includegraphics[width=0.98\linewidth, height = 0.2\textheight, keepaspectratio=true]{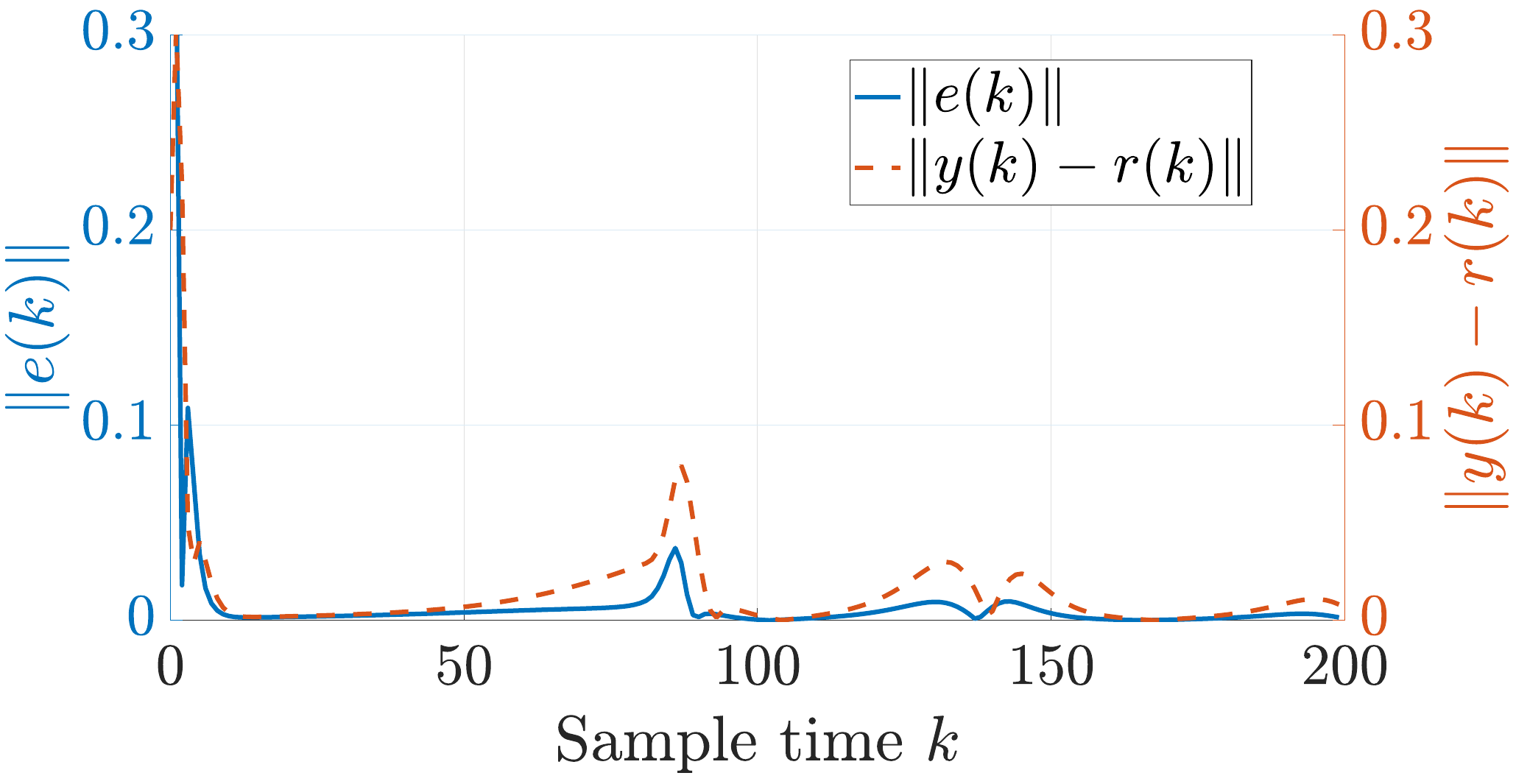}
            \end{minipage}} 
        \caption{CSTR plant. Closed-loop tracking of a generic reference trajectory using different disturbance models.}
        \label{fig:CSTR_gen}
\end{figure*}

This section contains numerical results showing the use of the proposed disturbance model for offset-free tracking of generic references.
We start with an academic example, where we show that a proper (but generally impractical) choice of the disturbance model leads to a prediction error $e(k)$ that vanishes to~$0$, leading to perfect offset-free tracking.
We then show a more realistic case study, where we showcase the good practical performance of the proposed disturbance scheme when compared to the classical constant disturbance model.
The results highlight the practical merits of the theoretical results presented in Section~\ref{sec:problem}, in that good offset-free results can be obtained even if the satisfaction of Assumptions~\ref{ass:perfect-modeling},~\ref{ass:contr} and~\eqref{eq:e(k)_asympt} is not guaranteed \emph{a priori}.

The results shown in this section use the MATLAB interface of CasADi \cite[version \texttt{3.6.0}]{casadi}, using IDAS \cite{IDAS} as the numerical integrator and IPOPT \cite{IPOPT} as the optimization~solver.

\subsubsection{Van der Pol oscillator} \label{sec:VanderPol}

We consider the Van der Pol oscillator system, whose dynamics are governed by the second-order Ordinay Differential Equation (ODE)
\begin{equation} \label{eq:VanderPol_cont}
    \frac{d^2 v}{d t^2} = \mu (1 - \beta v^2) \frac{d v}{d t} - v - \rho u,
\end{equation}
where $v \in \R$ is its position coordinate, $u \in \R$ its control input, and scalars $\mu = \beta = \rho = 1$ are its parameters.
We rewrite~\eqref{eq:VanderPol_cont} as an ODE $\frac{d x_\p}{d t} = F_\p(x_\p, u)$ by taking $x_\p = (\frac{d v}{d t}, v)$, and obtain \eqref{eq:plant} by taking $y_\p = v$ as the output and numerically integrating $F_\p(\cdot)$ with a sample time of $t_s = 0.5$ seconds, i.e., 
\begin{equation} \label{eq:integration:plant}
    x_\p(k+1) = \int_0^{t_s} F_\p(\xi(t), u(k)) d t, \; {\rm s.t.}\, \xi(0) = x_\p(k).
\end{equation}
To obtain the prediction model \eqref{eq:model:dxdy}, which is the particularization of model~\eqref{eq:model} for its use with the EKF, we first take $\tilde{F}_\p(\cdot)$ as a version of function $F_\p(\cdot)$ where the system parameters have been innecuratelly estimated to be $\mu = 0.8$, $\beta = 0.9$ and $\rho = 0.8$, and then numerically integrate $\tilde{F}_\p(\cdot) + h_x(x, u, \theta)$ as in~\eqref{eq:integration:plant}.
We consider three disturbance models:
\begin{enumerate}[label=(\textit{\roman*})]
    \item \textit{Constant Disturbance Model} (CDM).
        We take $y = v + d_y$, $h_x = 0$, $h_y = \theta$.
        This is the classical output disturbance model capable of achieving offset-free tracking for piecewise-constant references \cite{pannocchia2015, morari2012}.
        We include this model to show the limitations of this approach when considering non-constant reference trajectories.
    \item \textit{Polynomial Disturbance Model} (PDM).
        We take $h_y = 0$, and $h_x(x, u, \theta)= (\dot{v}_\theta, 0)$, where $\dot{v}_\theta$ is as a polynomial with terms capturing all possible terms of the ordinary differential equation of the plant, i.e., of \eqref{eq:VanderPol_cont}.
        That is, denoting $\dot{v} \doteq \frac{d v}{d t}$ and the $i$-th element of $\theta \in \R^{10}$ as $\theta_i$,
        \begin{align*}
            \dot{v}_\theta &= \theta_1 + \theta_2 \dot{v} + \theta_3 \dot{v}^2 + \theta_4 v + \theta_5 v^2 \\ 
                              &\quad + \theta_6 \dot{v} v + \theta_7 \dot{v}^2 v + \theta_8 \dot{v} v^2 + \theta_9 \dot{v}^2 v^2 + \theta_{10} u.
        \end{align*}
        The idea of this disturbance model is to guarantee the existence of a value of $\theta$ such that the prediction model \eqref{eq:model:dxdy} perfectly captures the real system dynamics \eqref{eq:VanderPol_cont}.
        Therefore, this disturbance model guarantees the satisfaction of Assumption~\ref{ass:perfect-modeling}.
        We include it to show how offset-free tracking may be achieved by a suitable selection of the disturbance model using the EKF and NMPC framework, although in a real setting the availability of such a disturbance model is unlikely.
    \item \textit{Feedforward Neural Network} (FNN).
        We take $h_x$ as a FNN with input $(x, u) \in \R^{\nx + \nu}$, output $d_x \in \R^\nx$, two hidden layers with $6$ neurons each, sigmoid activation function for the hidden layers, and linear activation function for the output layer.
        We also take $h_y$ as a similar FNN, but with input $x \in \R^\nx$, output $d_y \in \R^\ny$ and a single hidden layer with $4$ neurons.
        The parameters of the two FNNs, i.e., the weights and bias terms of the layers, are stacked in vectors $\theta_x$ and $\theta_y$, such that $\theta = (\theta_x, \theta_y)$.
        We initialize $\theta_x$ and $\theta_y$ using the Xavier initialization procedure \cite{Xavier2010} with zero bias terms.
        Inspired by \cite{bemporad_TAC_2023}, this disturbance model is essentially a recurrent neural network (RNN) that is trained online using the EKF to capture the discrepancy between the dynamics of the prediction model and the real system.
        We include this disturbance model to show the good tracking performance that can be obtained by using a FNN, which is a reasonable choice in a practical setting where there is limited information about the real system dynamics.
\end{enumerate}

We construct the NMPC controller \eqref{eq:nmpc} with a prediction horizon $N = 5$, taking a terminal equality constraint $x_N = x_\r(k+N)$ and using the classical stage cost function
\begin{align*}
    \ell(x, u, x_r, u_r) &= \| x - x_r \|^2_{W_x} + \| u - u_r \|^2_{W_u},
\end{align*}
with $W_x = 10 I_\nx$ and $W_u = I_\nu$.
The terminal cost $V_\mathrm{f}$ is not included because it is not needed due to the use of the terminal equality constraint.
We don't consider any constraints on the input nor on the output of the system.
At each sample time, we solve problem~\eqref{eq:ref_signals} taking a prediction horizon of length $N$ (instead of $\infty$), to obtain the future $N$ samples of the reference trajectories $x_\r(k+j)$ and $u_\r(k+j)$ for the NMPC \eqref{eq:nmpc}.
We take the stage cost function of~\eqref{eq:ref_signals} as~\eqref{eq:ref_signals:cost}, taking $u_d(k) = u_\r(k)$, where $u_\r(k)$ is the reference signal given by Assumption~\ref{ass:trackability}.

{
\setlength{\tabcolsep}{3pt}
    \begin{table}[t]
    \centering
    \begin{tabular}{c|ccccccc}
        \toprule
        Test & Dist. model & $\ndx$ & $\ndy$ & $\nth$ & $Q_x$ & $Q_y$ & $Q_\theta$ \\
        \midrule
         & CDM  & 0 & 1 & 1 & $I_\nx$ & $0.25 I_\ny$ & $I_\nth$ \\
       Fig.~\ref{fig:VanderPol_ss},~(\ref{fig:CSTR_gen})
         & PDM  & 2, (2) & 0 & 10, (7) & $I_\nx$ & $0.25 I_\ny$ & $I_\nth$ \\
         & FNN & 2 & 1 & 97 & $I_\nx$ & $0.25 I_\ny$ & $I_\nth$ \\
        \midrule
         & CDM  & 0 & 1 & 1 & $I_\nx$ & $0.25 I_\ny$ & $I_\nth$ \\
        Fig.~\ref{fig:VanderPol_gen} & PDM  & 2 & 0 & 10 & $10^{-10} I_\nx$ & $0.25 I_\ny$ & $50 I_\nth$ \\
                                     & FNN & 2 & 1 & 97 & $10^{-10} I_\nx$ & $0.25 I_\ny$ & $50 I_\nth$ \\
        \bottomrule
    \end{tabular}
    \caption{Dimensions and EKF parameters of each test.}
    \label{tab:EKF}
\end{table}
}

We perform two tests for each of the above disturbance models, one using a piecewise-constant reference and another for a generic reference trajectory.
The results are shown, respectively, in Fig.~\ref{fig:VanderPol_ss} and \ref{fig:VanderPol_gen}.
Table~\ref{tab:EKF} shows the parameters of the EKFs~\eqref{eq:ekf} used in each of the tests.
Fig.~\ref{fig:VanderPol_ss} shows that all three disturbance models achieve offset-free tracking of piecewise-constant references.
This is a well-known result in the case of the CDM \cite{pannocchia2015}, which does not hold for non-constant reference trajectories as shown in Fig.~\ref{fig:VanderPol_gen:const}.
On the other hand, Fig.~\ref{fig:VanderPol_gen:poly} shows how the proposed disturbance model is capable of achieving offset-free tracking if Assumptions~\ref{ass:trackability}--\ref{ass:contr} are satisfied in practice and $e(k) \to 0$.
Indeed, the parameters $\theta$ of the disturbance model converge to $\theta = (0, 0.2, 0, 0, 0, 0, 0, -0.28, 0, 0.2)$,
which is the value for which the prediction model \eqref{eq:model} is equivalent to the real plant~\eqref{eq:plant}.
We note that in a real setting, a polynomial disturbance model capable of capturing the exact system dynamics will generally not be available.
In this case, the use of a more general disturbance model, such as a FNN, still provides better reference tracking than a simple CDM, as illustrated in Fig.~\ref{fig:VanderPol_gen:FNN}.

\subsubsection{Continuous stirred tank reactor} \label{sec:results:CSTR}

We now consider the Continuous Stirred Tank Reactor (CSTR) system from the MPC toolbox for MATLAB~\cite{MPC_toolbox}.
The two states of the system are the temperature of the reactor $T_r$ and the concentration $C_A$ of the reactant $A$, the input is the temperature of the coolant $T_c$ and the output is $C_A$.
The control objective is to make $C_A$ track a given reference trajectory.
We obtain a discrete-time plant model~\eqref{eq:plant} by integrating its ODE as in~\eqref{eq:integration:plant} with $t_s = 0.5$ seconds.
As in Section~\ref{sec:VanderPol}, we take the prediction model \eqref{eq:model:dxdy} by changing some of the parameters of the plant model~\eqref{eq:plant}.
We consider the same disturbance models described in Section~\ref{sec:VanderPol}, although in this case the PDM is taken as a copy of the plant model \eqref{eq:plant} with parameters collected in $\theta$.
We also take the same parameters for the reference generator \eqref{eq:ref_signals} and NMPC controller \eqref{eq:nmpc}, with the exception of $W_x$, which we take as $W_x = \texttt{diag}(1, 0.1)$.

Fig.~\ref{fig:CSTR_gen} shows the closed-loop results of the CSTR system tracking a generic reference trajectory.
The parameters of the EFK are shown in Table~\ref{tab:EKF}.
Fig.~\ref{fig:CSTR_gen:poly} shows that the use of a simple disturbance model designed to capture the discrepancy between the real plant and the prediction model can lead to near-perfect reference tracking.
The CDM provides good results when the reference changes slowly, as seen in the first $50$ sample times of Fig.~\ref{fig:CSTR_gen:const}.
However, its performance degrades significantly otherwise.
Finally, once again, Fig.~\ref{fig:CSTR_gen:FNN} shows how a FNN can provide very good tracking results.

\vspace*{-0.1em}
\section{Conclusions and future research directions} \label{sec:no_guarantees}

We have presented sufficient conditions for offset-free tracking of time-varying references when considering a nonlinear system with unknown dynamics controlled by a nonlinear controller and observer.
Offset-free tracking is achieved by considering a nonlinear disturbance model whose parameters are learned online by the observer. By combining a reasonably-designed NMPC with EKF for state and parameter estimation, we have shown that good tracking results can be obtained even if all the theoretical assumptions are not fully satisfied.

Many research issues have been left open that require future investigations.
First, in order to guarantee offset-free tracking, it would be important to characterize EKF and NMPC schemes that, contrarily to the ones we have used as described in Section~\ref{sec:NMPC_EKF}, guarantee the satisfaction of all the assumptions presented throughout Section~\ref{sec:problem}; in particular an NMPC scheme that is robust to vanishing perturbations
in spite of the time-varying nature of the reference and prediction model. 
Second, one should propose disturbance models that satisfy Assumption~\ref{ass:perfect-modeling} by design.
Finally, some of the assumptions made in~Section~\ref{sec:problem} could be relaxed; in particular,
the assumption of convergence to zero of the prediction error made in Theorem~\ref{th:main} 
may be difficult to guarantee, even if Assumptions~\ref{ass:trackability}-\ref{ass:contr} are satisfied,
due to the lack, in general, of a separation principle between NMPC and EKF design.

\vspace*{-0.1em}
\bibliographystyle{ieeetr}
\bibliography{bibliography}

\end{document}